\documentclass[%
  reprint,
superscriptaddress,
 amsmath,amssymb,
 aps,
usenames,dvipsnames
]{revtex4-2}

\usepackage{graphicx}%
\usepackage{dcolumn}%
\usepackage{bm}%

\usepackage{tikz}
\usetikzlibrary{quantikz}

\usepackage{amsthm}

\newtheorem{lemma}{Lemma}

\newcommand{\la}{\langle}
\newcommand{\ra}{\rangle}

\newcommand{\aAb}{|\la a|A|b \ra|^2}
\renewcommand{\aa}{\acute{a}}
\newcommand{\bb}{\grave{b}}
\renewcommand{\rm}[1]{\textrm{#1}}

\newcommand{\half}{\frac{1}{2}}
\newcommand{\ep}[1]{^{(#1)}}

\newcommand{\mO}{\mathcal O}
\newcommand{\mS}{\mathcal S}
\newcommand{\m}[1]{\mathcal{#1}}
\newcommand{\ac}[1]{\acute{#1}}
\newcommand{\g}[1]{\grave{#1}}

\newcommand{\Aloc}{A_{\text{loc}}}
\newcommand{\Anonloc}{A_{\text{nonloc}}}

\newcommand{\supr}[1]{^{(#1)}}

\newcommand{\notrap}{NOTraP}

\begin{document}

\preprint{APS/123-QED}

\title{Improved resource-tunable near-term quantum algorithms for transition probabilities, with applications in physics and variational quantum linear algebra}%

\author{Nicolas PD Sawaya}
\email{nicolas.sawaya@intel.com}
\affiliation{%
Intel Labs, Santa Clara, California 95054, USA \\
}%

\author{Joonsuk Huh}%
\affiliation{Department of Chemistry, Sungkyunkwan University, Suwon 16419, Korea}
\affiliation{SKKU Advanced Institute of Nanotechnology (SAINT), Sungkyunkwan University, Suwon 16419, Korea}
\affiliation{Institute of Quantum Biophysics, Sungkyunkwan University, Suwon 16419, Korea}

\begin{abstract}
Transition amplitudes and transition probabilities are relevant to many areas of physics simulation, including the calculation of response properties and correlation functions. These quantities can also be related to solving linear systems of equations. 
Here we present three related algorithms for calculating transition probabilities. First, we extend a previously published short-depth algorithm, allowing for the two input states to be non-orthogonal. Building on this first procedure, we then derive a higher-depth algorithm based on Trotterization and Richardson extrapolation that requires fewer circuit evaluations. Third, we introduce a tunable algorithm that allows for trading off circuit depth and measurement complexity, yielding an algorithm that can be tailored to specific hardware characteristics. Finally, we implement proof-of-principle numerics for models in physics and chemistry and for a subroutine in variational quantum linear solving (VQLS). The primary benefits of our approaches are that (a) arbitrary non-orthogonal states may now be used with small increases in quantum resources, (b) we (like another recently proposed method) entirely avoid subroutines such as the Hadamard test that may require three-qubit gates to be decomposed, and (c) in some cases fewer quantum circuit evaluations are required as compared to the previous state-of-the-art in NISQ algorithms for transition probabilities.

\end{abstract}

\maketitle

\section{Introduction}

If breakthroughs in hardware design continue, quantum computers will be able to simulate quantum systems that are classically intractable, including those from condensed matter physics \cite{wecker15}, nuclear structure \cite{dumitrescu18}, high-energy physics \cite{bauer19_hep}, and chemistry and materials \cite{cao19_rev,mcardle20_rev}.
As theoretical and algorithmic work progresses, it is imperative to continue improving quantum computational primitives and subroutines so that calculations can be available on the earliest possible hardware.

One important subroutine is the calculation of transition amplitudes and probabilities, quantities closely related to the Fermi golden rule, response functions, and correlation functions more generally. These are required for calculating intensities in various areas of spectroscopy \cite{huh15,sawaya19_vibronic,kosugi20_linresp_chargespin,cai20_qmolecresp,ibe20_transition,sawaya2021ir} and for response functions in scattering experiments and condensed matter \cite{berne70_corrbook,florencio2020corrrev,Roggero19_linresp,chen2021vardyncorrfuncs}. Additionally, as vector-matrix-vector products $\vec a^t \textbf{A} \vec b$ (or $\vec a^t \textbf{A} \vec a$) are often relevant to classical linear algebra problems, transition probability subroutines may be useful in quantum linear algebra \cite{bravo19_vqls,gilyen2019qsvd,xu19_vari_linalg,endo20_vari_genproc,wang2021vqsvd,huang2021near}, including for classical partial differential equations \cite{berry2014pde,arrazola2019pde,liu2021pde}, finance \cite{orus2019finance,pistoia2021jpmc_finance}, and quantum machine learning \cite{liu20_qml,biamonte2017quantum,garcia2022qmlrev}.

In this work we introduce three algorithms for calculating transition probabilities. We refer to our algorithm framework as NOTraP, for non-orthogonal transition probabilities.

First, we show how to extend an existing short-depth quantum subroutine to allow for non-orthogonal states, at resource costs of just one additional qubit and at most four additional two-qubit gates. This short-depth algorithm  (NOTraP-SD) allows for a broader class of states and applications. As with previous related methods \cite{ibe20_transition,bravo19_vqls}, an important feature is that controlled unitaries (meaning known state preparation unitaries and unitaries for Pauli string exponentials) are not required for any algorithms in this work.

Second, we demonstrate a higher-depth method (NOTraP-HD) that greatly reduces the number of distinct quantum circuits that must be simulated. Third, we show how one may tune (NOTraP-T) between low circuit depths and a low number of distinct quantum circuits. This allows one to tailor the algorithm to a given set of hardware. For example, if one is given access to a quantum computer allowing for larger maximum circuit depth than before, one may modify the algorithm to use deeper circuits with the trade-off of having fewer circuit evaluations.

There are many problems for which states $|a\ra$ and $|b\ra$ are non-orthogonal, for example the calculation of Franck-Condon factors in molecules \cite{huh15,sawaya19_vibronic}, arbitrary correlation functions in condensed matter systems, or comparing two states in classical linear algebra. Other than this expansion of the types of states that may be used as input, the algorithms of this work also provide the benefits of tunable circuit depths and general applicability to linear algebra.

This work should be considered in the context of recent efforts to derive algorithms that eliminate the use of controlled-unitary circuits \cite{mitarai19_newhadamtest,ibe20_transition,huggins20_nonorthogonal,bravo19_vqls,huang2022linearoptical}, especially the Hadamard test. Since the native gate set of many quantum computers will consist solely of one- and two-qubit gates, such controlled state preparations require many three-qubit gates to be decomposed into simpler gates. Controlled state preparation unitaries thus would have lead to circuits several times deeper, making algorithms with many (decomposed) three-qubit gates prohibitive on most near-term quantum hardware.

\begin{figure}
    \centering
    \includegraphics[width=0.45\textwidth]{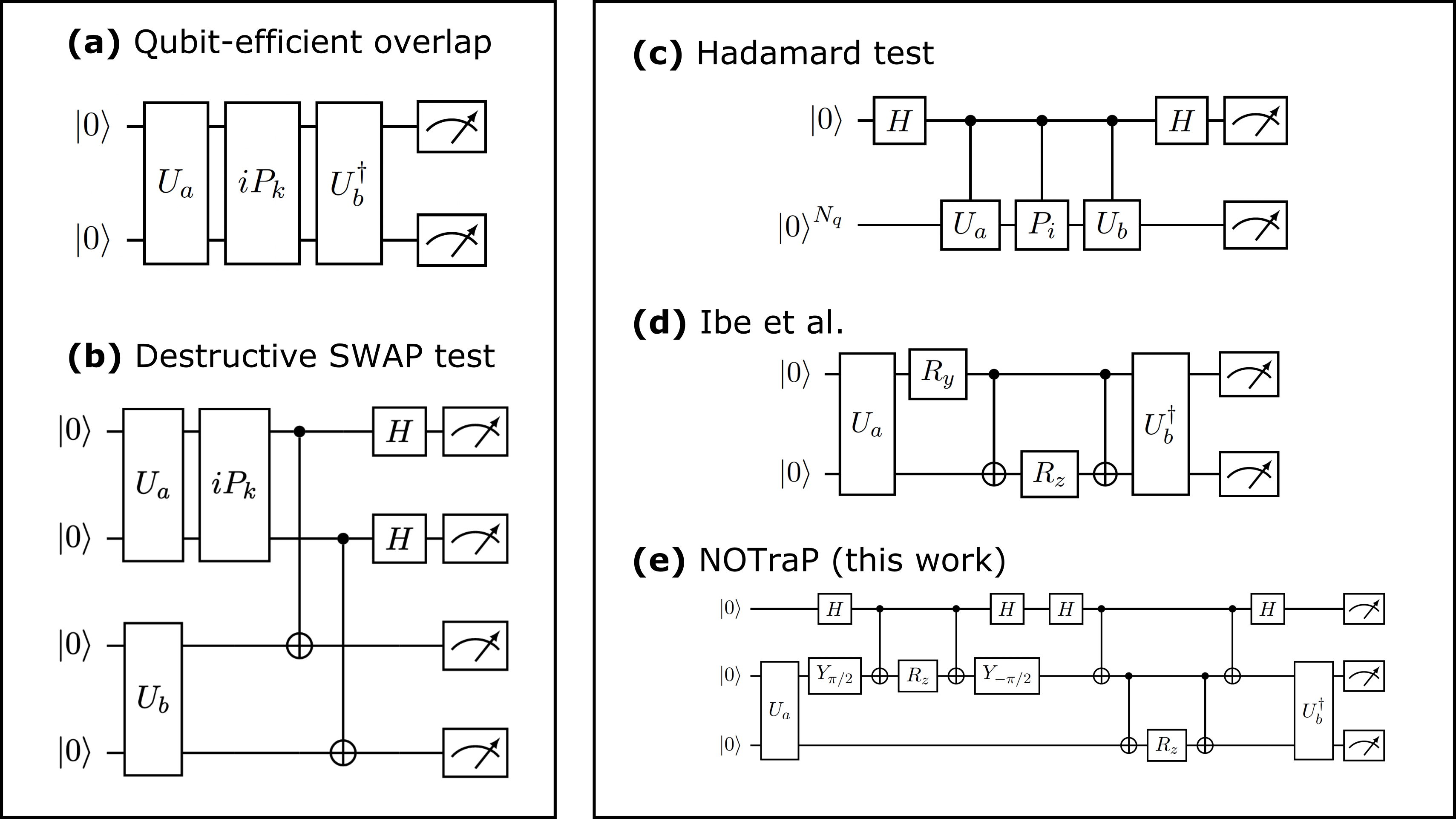}
    \caption{Left panel: Circuits for calculating quantities $|\la a|i P_k|b \ra|^2$ with respect to a single Pauli string $P_k$. 
    (a) Qubit-efficient overlap circuit, given state preparation unitaries $U_a|0\ra=|a\ra$ and $U_b|0\ra=|b\ra$. 
    (b) Use of the shorter-depth destructive SWAP test \cite{garcia13_destrswap}, requiring twice the qubits.
    Right panel: Circuits indirectly used in determining cross-terms between Pauli strings, using different methods. For illustrative purposes the bottom two circuits implement the cross-terms in operators
    $A = k_1 Y_1 + k_2 Z_1 Z_2$ and $\dot A = k_1 X_0 Y_1 + k_2 X_0 Z_1 Z_2$. 
    (c) The Hadamard test for determining $\Re \la a|P_i|b \ra$; a similar circuit yields $\Im \la a|P_i|b \ra$ \cite{knill07_hadtest,dobsicek07_hadtest}. Because the Hadamard test retains phase information, cross-terms $\la a|P_i|b \ra \la b|P_j|a \ra$ may then be determined classically via products between $\la a|P_i|b \ra$ and $\la b|P_j|a \ra$. This circuit uses controlled state preparation unitaries, and therefore is considerably deeper if gate decomposition to one- and two-qubit gates is required.  
    (d) Circuits used by Ibe \textit{et al.} \cite{ibe20_transition} to calculate terms $W_2$, $W_3$, and $W_4$ of equation \eqref{eq:w1234}. This method requires that $|a\ra$ and $|b\ra$ be orthogonal. 
    (e) Circuits for implementing NoTraP-SD, after transforming $\{A,|a\ra,|b\ra\} \rightarrow \{\dot A,|\aa\ra,|\bb\ra\}$ via equations \eqref{eq:abxform} and \eqref{eq:Adot}. Importantly, NoTraP allows for transition probabilities to be calculated (i) for arbitrary non-orthogonal states $|a\ra$ and $|b\ra$, and (ii) without controlled state preparation unitaries nor controlled Pauli rotations. }
    \label{fig:cross-circs}
\end{figure}

\begin{table}[h!]
\centering
\begin{tabular}{c c c c} 
 \hline
 Method & Circ Evals & $\m D[A_{loc}$]  & $\m D[A_{nonloc}$]  \\ [1ex]
 \hline
Hadam. test \cite{knill07_hadtest,chen2021vardyncorrfuncs}  &
$\m O(n_q^2)$ & $\m O(n_{\rm{Toff}})^*$ & $\m O(n_{\rm{Toff}} n_q)^*$ \\ 
original VQLS \cite{bravo19_vqls} & 
$\m O(n_q^2)$ & $\m O(n_{\rm{Toff}})$ & $\m O(n_{\rm{Toff}} n_q)$ \\
\notrap-SD \& Ibe \cite{ibe20_transition} &
$\m O(n_q^2)$ & $\m O(1)$ & $\m O(n_q)$ \\
\notrap-HD &
$\m O(n_{\tau} )$ & $\m O(n_q)$ & $\m O(n_q^2)$ \\
\notrap-T &
$\m O( \frac{n_\tau n_q^2}{N_G^2} )$ & $\m O( \frac{n_\tau n_q}{N_G} )$ & $\m O( \frac{n_\tau n_q^2}{N_G} )$ \\ [1ex] 
 \hline
\end{tabular}
\caption{Circuit depths and distinct quantum circuit evaluations for calculating transition probabilities with respect to $A_{loc}$ and $A_{nonloc}$. The asterisk ($^*$) denotes methods that require three-qubit gates and hence longer depths upon decomposition. While this table does not consider depth for the state preparation unitaries of $|a\ra$ and $|b\ra$, the Hadamard test is the only method in the table that requires controlled versions of the state preparation circuits for $|a\ra$ and $|b\ra$, greatly increasing the overall circuit depth on most current hardware. 
}
\label{tbl:meas_depth}
\end{table}

\section{Theory}\label{sec:theory}

The quantity of interest is the transition probability
\begin{equation}\label{eq:transamp}
|\la a|A|b \ra|^2
\end{equation}
where $|a\ra$ and $|b\ra$ are arbitrary states and $A$ is a Hermitian operator that may be expressed as
\begin{equation}\label{eq:Apaulidecomp}
A = \sum_k^{N_P} g_k P_k
\end{equation}
where $g_k$ is a real constant, $P_k$ is a Pauli string $P_k \in \{I,X,Y,Z\}^{\otimes n_q}$, $N_P$ is the number of Pauli strings in the decomposition, $n_q$ is the number of qubits, $I$ is the identity and $\{X,Y,Z\}$ are the Pauli matrices. As previously mentioned, $A$ may come from a quantum problem (such as chemistry) or may be the quantum representation of a matrix from a classical problem. In the latter case, the goal may be to determine formula \eqref{eq:transamp} in order to solve a linear system variationally \cite{bravo19_vqls,xu19_vari_linalg,endo20_vari_genproc,wang2021vqsvd,huang2021near}.

In order to study scaling behavior of local and nonlocal operators, we will consider the following two simple operators. On $n_q$ qubits, they are the operators %

\begin{equation}
A_{loc} = \sum_{i=0}^{n_q-1} X_i / \sqrt{n_q}
\end{equation}
and
\begin{equation}
\begin{split}
A_{nonloc} &= \sum_{i=0}^{n_q-1} X^{\otimes n_q} X_i  / \sqrt{n_q}\\
&= \sum_{i=0}^{n_q-1} X^{\otimes i} I_i X^{\otimes n_q-i-1}   / \sqrt{n_q}
\end{split}
\end{equation}

where $X_i$ is the Pauli-X operator on qubit $i$, $I$ is the identity, and $n_q$ is the number of qubits. Considering these two operators will allow us to study some general trends in scaling, including circuit depths and number of distinct circuit evaluations.

\subsection{Previous methods}\label{sec:prevmethods}
This subsection reviews the following previously reported methods for transition probabilities: the use of the Hadamard test \cite{knill07_hadtest,chen2021vardyncorrfuncs} to determine each individual cross-term between Pauli strings of \eqref{eq:Apaulidecomp}, a related approach that does not require controlled unitaries andinstead uses short-depth controlled Pauli strings \cite{bravo19_vqls}, and the method of Ibe \textit{et al.} \cite{ibe20_transition} that (assuming state orthogonality) eliminated the need for the Hadamard test in this context. We focus primarily on the latter method as it is similar in spirit to the approach of this work. Circuits for each method are shown in Figure \ref{fig:cross-circs}.

We begin by summarizing how to use the more expensive Hadamard test. 
We define unitaries $U_a|0\ra = |a\ra$ and $U_b|0\ra = |b\ra$, where $|0\ra$ is some reference state. Established near-term algorithms exist for determining these circuit unitaries for states of interest \cite{mcardle19_ite,ollitrault19_eom,mcclean16_njp,jones19_discovspectra,higgott19_vqd,motta19_qite,mclean17_qse,santagati18_waves}. In this work we assume that these state preparation unitaries are already known.

An obstacle in calculating quantity $\aAb$ is that $A$ is Hermitian but not unitary, meaning it cannot be trivially implemented with a unitary quantum circuit (though the longer-depth linear combination of unitaries appoach may be used \cite{childs2012lcu}). However, each individual unitary $P_k$ may be implemented up to global phase with a simple circuit \cite{mikeike11}. 
Furthermore, an expansion of equation \eqref{eq:transamp} reveals cross-terms such as $g_i g_j \la a|P_i|b \ra \la b|P_j|a \ra$, for which a naive implementation would use the Hadamard test \cite{knill07_hadtest,dobsicek07_hadtest,mitarai19_newhadamtest}.

Using the Hadamard test one may determine $\la a|P_i|b \ra$ by using a controlled state preparation unitary $cU$, where the the unitary to be controlled is $U=U_a^\dag P_i U_b$. Even though this arguably could be classified as a near-term method, its depth requirements are much larger than the original unitary. Assuming the original unitary was composed only of one- and two-qubit gates, these respectively become two- and three-qubit gates, which in turn have costly decompositions into smaller gates.
Hence a major purpose of previous studies \cite{ibe20_transition,huggins20_nonorthogonal} and of the current work is to circumvent the Hadamard test.

In variational quantum linear solving (VQLS) \cite{bravo19_vqls}, as originally proposed, an ancilla control qubit is used to implement controlled $P_k$ operators while the state preparation unitaries are implemented directly without control qubits. These controlled $P_k$ operations can be directly implemented with two-qubit gates, leading to a substantial reduction in circuit depth.

A notable advance was recently introduced by Ibe et al. \cite{ibe20_transition}. Without needing any ancilla qubits, their method uses measurements from many distinct circuits, before taking a weighted sum of each result to arrive at the desired quantity. We summarize their method here using slightly modified notation.

The algorithm of Ibe \textit{et al.} begins by calculating the following quantities on a quantum computer, each of which is in the form $|\la a|V_i|b \ra|^2$ where $V_i$ is a different unitary. We define quantities
\begin{equation}\label{eq:w1234}
\begin{split}
W_1\ep{k} = |\la a| P_k |b \ra|^2 \\
W_2\ep{kl} = |\la a| \half(I+iP_k)(I+iP_l) |b \ra|^2 \\
W_3\ep{kl} = |\la a| \half(I-iP_k)(I-iP_l) |b \ra|^2 \\
W_4\ep{kl} = |\la a| P_k P_l |b \ra|^2, \\
\end{split}
\end{equation}
each of which is a variable between 0 and 1. %

Unitaries involving $P_k$ and $P_l$ may be implemented using known circuits for exponentiating a single Pauli string \cite{mikeike11}, while the states themselves are prepared using $U_a$ and $U_b$. Calculating the overlaps squared is possible using the SWAP test \cite{buhrman01_qfingerprint} or destructive SWAP test \cite{garcia13_destrswap} between states $U_a|0\ra$ and $V U_b|0\ra$, or by implementing $U_a^\dag V U_b$ before determining the frequency of measured state $|0\ra^{\otimes n_q}$ \cite{havlek19_qml}, as shown in see Figure \ref{fig:cross-circs}(a) and (b). The former methods require $2n_q$ qubits and the latter requires $n_q$ qubits and at most a doubling of circuit depth.

The transition probability may then be reconstructed as
\begin{equation}\label{eq:aAb_sq_1}
\begin{split}
|\la a| A |b \ra|^2 = \sum_k g_k^2 W_1\ep{k} \\
+ \sum_k \sum_{l<k} g_k g_l \big [  2 W_2\ep{kl} + 2 W_3\ep{kl} - W_1\ep{k} - W_1\ep{l} - W_4\ep{kl}  \big ]
\end{split}
\end{equation}

or equivalently %
\begin{equation}\label{eq:aAb_sq_2}
\begin{split}
|\la a| A |b \ra|^2 = \\
\sum_k (g_k^2 - 2\sum_{l<k} g_k g_l) W_1\ep{k}  \\
+ \sum_k \sum_{l<k} g_k g_l \big ( 2 W_2\ep{kl} + 2 W_3\ep{kl} - W_4\ep{kl} \big ) \\
\end{split}
\end{equation}

where the equality holds only if $\la a|b \ra=0$, though below we propose a method usable for non-orthogonal states. The number of distinct circuits is $(3N_P^2-N_P)/2$.
The number of measurements required for determining each of formulas \eqref{eq:w1234} is $\mO(1/\epsilon_i^2)$, where $\epsilon_i$ is the required precision for each circuit; a full analysis of measurement counts is slightly more involved and is derived in the next subsection.

Throughout this work we assume the available gateset is CNOT and all arbitrary one-qubit unitaries. The depths will change if a different native gate set is assumed. The maximum depth (excluding the state preparation circuits) for the Ibe \textit{et al.} circuit is 
\begin{equation}\label{eq:}
D^{\rm{Ibe}} = 5 + 4(k-1) = 4k + 1
\end{equation}
where $k$ is the longest Pauli string length in the Hamiltonian. This upper bound comes from having a circuit $e^{-i\frac{\pi}{4}P_k}e^{-i\frac{\pi}{4}P_l}$, which in the worst case operate on the same set of qubits. This leads to 4 CNOT-ladders each with $k-1$ CNOT gates, two R$_z(\frac{\pi}{4})$ rotations, and 4 layers of basis changes of which the middle two layers may be combine into one layer of single qubit gates. This in turn yields $D^{\rm{Ibe}}[A_{nonloc}] = 4n_q$ and $D^{\rm{Ibe}}[A_{loc}] = 8$.

For the circuits of the original VQLS formalism \cite{bravo19_vqls}, based on controlled Pauli rotations (cPR), we consider controlled versions of the circuits of the previous paragraph, which leads to deeper circuits after decomposition. 
In this case we have the previous CNOT gates replaced with Toffoli gates, the R$_z(\frac{\pi}{4})$ gates replaced with controlled-R$_z$, and the depth-1 layer of basis change replaced by a depth-$k$ layer of controlled rotations. The maximum depth for VQLS is 
\begin{equation}
D^{\rm{VQLS}} = 4 (k-1) (11) + 4 k (4)
\end{equation}
where we use $D_{\rm{Toff}}=11$ for the depth required to decompose a Toffoli (controlled CNOT) gate and $D_{\rm{cR}}=4$ for depth of a decomposed arbitrary controlled one-qubit rotation \cite{barenco95decomp}. This leads to $D^{\rm{VQLS}}[A_{nonloc}] = 60 n_q - 104$ and (performing additional gate cancellations via inspection) $D^{\rm{VQLS}}[A_{loc}] = 8$.

\subsection{Shot (measurement) counts for previously published methods}

To our knowledge, a resource analysis for required shot counts (i.e. circuit repetitions) has not been performed for the previously published methods. Here we perform such an analysis for Ibe \textit{et al.} \cite{ibe20_transition} and for the original VQLS \cite{bravo19_vqls} formalism. 

We begin with an analysis of Ibe \textit{et al.} The number of measurements required for the full calculation of $Q=\aAb$ is dependent on the coefficients in equation \eqref{eq:aAb_sq_2}. There are $N_{W}$ circuits from which measurements need to be extracted, hence $N_W$ independent variables. Re-indexing all random variables $W$ with $i$, if the goal is an upper additive error bound of $\varepsilon_Q$, then via standard error propagation we have
\begin{equation}
\varepsilon_Q^2 = \sum_i^{N_W} \left |\frac{\partial Q}{\partial W_i} \right|^2 \text{Var}(W_i) \varepsilon_{i}^2.
\end{equation}
where $\text{Var}(W_i)$ is the variance of sampling circuit $W_i$ and $0 \leq \text{Var}(W_i) \leq 1$ \cite{crawford2021si,yen2023deterministic}. In order to obtain analytical results, for the rest of this section we will calculate upper bounds by assigning $\text{Var}(W_i)=1$.
Note that $\varepsilon_Q$ has the same units as $A^2$, while $\varepsilon_i$ is dimensionless as it is simply the error in $W_i$. %

Assigning each summed term to have the same uncertainty leads to
\begin{equation}
\left |\frac{\partial Q}{\partial W_i} \right|^2 \varepsilon_{i}^2 = \varepsilon_Q^2 / N_W
\end{equation}

which in turn yields measurement counts of each term as 
\begin{equation}
n_i = \frac{1}{\varepsilon_i^2} = \frac{N_W\left |\frac{\partial Q}{\partial W_i} \right|^2}{\varepsilon_Q^2}
\end{equation}
and total measurements (shots)
\begin{equation}
n_{tot} = \sum_i n_i
\end{equation}

where

\begin{equation}\label{eq:NW}
N_W = N_P + \frac{3}{2} (N_P-1) N_P
\end{equation}

Using equation \eqref{eq:aAb_sq_2} the general result for arbitrary operator $A=\sum_k g_k P_k$ is
\begin{equation}
n_{tot} = \frac{N_W}{\varepsilon_Q^2} \left ( \sum_k \left (2g_k^2 - \sum_{l<k} g_k g_l \right )^2 + \sum_k \sum_{l<k}(9g_k^2g_l^2) \right ).
\end{equation}

It is instructive to consider the number of measurements for $A_{loc}$ and $A_{nonloc}$, in order to obtain an understanding of basic scaling. As all $g_k$ equal unity for these two operators, the measurement counts are 
\begin{equation}\label{eq:ntot_general}
\begin{split}
n_{tot}[A_{loc/nonloc}] = \\  \frac{N_W}{\varepsilon_Q^2}( \frac{1}{6}N_P(2N_P^2-15N_P+37) + \frac{9}{2} N_P(N_P-1)) \\
= \m O(N_P^5/\varepsilon_Q^2) = \m O(n_q^5/\varepsilon_Q^2)
\end{split}
\end{equation}
where we have used $N_W = \m O(N_P^2)$ from equation \eqref{eq:NW}. The perhaps unexpected fifth-order scaling is a result of the first sum of equations \eqref{eq:aAb_sq_2} and \eqref{eq:ntot_general} having a cubic scaling in $N_P$.

In this work we concern ourselves primarily with relative error $\eta = \frac{\varepsilon}{Q}$. Because $\aAb=\|A\|^2\frac{\aAb}{\|A\|^2}$, an order of magnitude estimate for the relative error $\frac{\varepsilon}{\aAb}$ is $\frac{\varepsilon}{\|A\|^2}$. As $\|\Aloc\|^2$ and $\|\Anonloc\|^2$ scale as $\m O(n_q^2)$, with respect to relative error the total number of measurements scales as $\m O(n_q^3/\eta_Q^2)$.

Now we do a measurement resource analysis for VQLS, for which the VQLS objective function is optimized when the correct solution $|x\ra$ to system of linear equations $A|x\ra=|b\ra$ is found. Unlike the Hadamard test, this approach does not require controlled state preparation unitaries, though it does require controlled Pauli terms.

Because cross-terms are calculated explicitly, the individual terms are
\begin{equation}
\sum_k g_k^2 [|\la a|P_k|b \ra|^2] + \sum_k \sum_{l<k} g_k g_l 2 \Re [\la a|P_l|b \ra\la b|P_k|a \ra]
\end{equation}
where each term in square brackets is a distinct circuit. The number of circuits is thus $N_W= N_P+(N_P-1)N_P/2$ and the measurement count is
\begin{equation}
n_{tot}^{\rm{VQLS}} = N_W \varepsilon_Q^{-2} \left[ \sum_k g_k^4 + \sum_k \sum_{l<k} 2 g_k^2 g_l^2 \right]
\end{equation}
which in turn leads to
\begin{equation}
\begin{split}
n_{tot}^{\rm{VQLS}}[A_{loc/nonloc}] &=  N_W \varepsilon_Q^{-2} \left[ N_P + (N_P-1) N_P \right] \\
&=  N_W \varepsilon_Q^{-2} \left[ N_P + (N_P-1) N_P \right]\\
\\
&= \m O(N_P^4) = \m O(n_q^4),
\end{split}
\end{equation}
corresponding to approximately $\m O(n_q^2/\eta_Q^2)$ measurements when considering relative error.

\section{Non-orthogonal states, extrapolation, and depth tuning}

In this section we introduce three new tools for calculating transition probabilities. First we extend Ibe \textit{et al.}'s previous method to allow for calculating transition probabilities between \textit{non-orthogonal} states, while still avoiding any gates higher than 2-qubit gates. We expect this method to be useful both for calculating transition probabilities between arbitrary quantum states and for more efficiently running VQLS. Second we show how to use exponentiation and extrapolation to greatly reduce the number of measurements required, at the cost of an increase in circuit depth. Finally, we demonstrate how one may tune between higher-depth and higher-measurement, in order to use as much of the limited (due to \textit{e.g.} noise) available circuit depth as possible when calculating transition probabilities.
Table \ref{tbl:meas_depth} show a comparison of the capabilities of our method compared to previous work.

\subsection{Implementing non-orthogonal states (NOTraP-SD)}\label{sec:nonorthog}

Our first algorithmic contribution is a simple modification to the states and to the operator $A$; after this modification, any procedure that would have required orthogonal input states can be used on the modified states. First, one adds a single ancilla qubit to the state space, with opposite bit values for each state of interest,
\begin{equation}\label{eq:abxform}
\begin{split}
|a\ra \rightarrow |0\ra_a \otimes |a\ra_s &\equiv |\ac a\ra \\
|b\ra \rightarrow |1\ra_a \otimes |b\ra_s &\equiv |\g b\ra
\end{split}
\end{equation}
Next one modifies the operator $A$, multiplying it by a bit-flip operator acting on the new qubit:
\begin{equation}\label{eq:Adot}
A \rightarrow X_a \otimes A_s \equiv \dot A.
\end{equation}
The result is that, even if an algorithm for formula \eqref{eq:transamp} requires that the two input states be orthogonal, one may use these new constructs ($|\ac a\ra, |\g b\ra$, and $\dot A$) to calculate the originally desired quantity $|\la a|A|b \ra|^2$. This is proven in the following two lemmas.

\begin{lemma}\label{lem:1}
$\la \ac a | \g b \ra = 0$. 
\end{lemma}
\begin{proof}

\begin{equation}
\begin{split}
&\la \ac a | \g b \ra \\
&= \la 0 | 1 \ra \la a | b \ra \\
&= 0
\end{split}
\end{equation}
\end{proof}

\begin{lemma}\label{lem:2}
$\la a|A|b \ra = \la \ac a | \dot A | \g b \ra$. 
\end{lemma}

\begin{proof}
\begin{equation}
\begin{split}
&\la \ac a|\dot A|\g b \ra \\
&= \la 0 | \la a | (X\otimes A) |1\ra | b \ra \\
&= \la 0 | X | 1 \ra \la a | A | b \ra \\
&= \la a | A | b \ra
\end{split}
\end{equation}
\end{proof}

This simple extension allows the exact transition probabilities for any two states to be calculated using the previously introduced \cite{ibe20_transition} formula \eqref{eq:aAb_sq_1}, which would otherwise produce the incorrect results if the two input states are non-orthogonal. The additional cost is one extra qubit, as well as \textit{at most} a circuit depth increase of four two-qubit entangling gates (sometimes just 2) and two one-qubit gates; note that adjacent Hadamard gates cancel in Figure \ref{fig:cross-circs}(e). %
Thus we still avoid the long depths required of the Hadamard test, but calculate the exact desired quantity. Allowing for non-orthogonal states is important for practical applications, with examples including variational linear systems solvers as well as simulating vibronic spectra.
This orthogonalization procedure is also necessary for the subsequent methods in this work.

\subsection{Exponentiation method (NOTraP-HD)}\label{sec:expmethod}

Here we develop another algorithm for calculating transition probabilities, with the goal of reducing the required measurement counts at the cost of an increased circuit depth. The method is based on exponentiating the operator $\dot A$ and using Richardson extrapolation. (Extraploation has been to improve efficiency in quantum linear algebra \cite{vazquez2022qla_extrap}, though this previous algorithm is more amenable to longer-term hardware and is not similar to our approach.) 
We begin by writing the Taylor expansion
\begin{equation}
e^{-i\tau \dot A} = \mathcal I - i \tau \dot A - \half \tau^2 \dot A^2 + \cdots
\end{equation}

We first require the determination of $| \la\aa|\exp(-i \tau \dot A)|\bb\ra |^2$ for multiple $\tau$. This quantity can be determined by the frequency of all-$|0\ra$ measurements from implementing the circuit for $U_{\bb}^\dag \exp(-i \tau \dot A) U_{\aa} |\mathbf 0\ra$. One may alternatively use the destructive SWAP test on states $\exp(-i \tau \dot A)U_{\aa} |\mathbf 0\ra$ and $U_{\bb}|\mathbf 0\ra$, which lowers the depth while doubling the number of qubits (see Figure \ref{fig:cross-circs}).

The algorithm proceeds by determining the following function for at least two values of $\tau$, before using Richardson extrapolation \cite{richardson1927deferred,pozrikidis2008} to determine $\aAb$. The function is
\begin{equation}\label{eq:f_expan}
\begin{split}
f(\tau) = |\la \ac a | e^{-i\tau \dot A} | \g b \ra|^2 + |\la \ac a | e^{+i\tau \dot A} | \g b \ra|^2 \\
= 2\tau^2 |\la  a |  A |  b \ra|^2 + \tau^4 K_4 + \tau^6 K_6 + \mathcal O(\tau^8)
\end{split}
\end{equation}
where many terms are combined into constants $K_n$. We note two key insights here. First, the quadratic term's coefficient is proportional to $\aAb$ only because of the transformations \eqref{eq:abxform} and \eqref{eq:Adot}; otherwise the quadratic term would have been contaminated with a term $\tau^2 \la b | A^2 | a \ra \la a | b \ra$. Hence the orthogonalization procedure is strictly necessary for this extrapolation-based algorithm. Second, the cancellation of all odd orders of $\tau$ in \eqref{eq:f_expan} is due to the summing of results from $+\tau$ and $-\tau$. %

An important consideration is the choice of $\tau$ values for the extrapolation. A smaller $\tau$ leads to a lower error for a given number of extrapolation points, but to a larger total number of required measurements.
The algorithm's dependence on $\tau$ is further discussed in the next subsection.

As we show presently, it is fortunately not necessary to exactly implement the exponential in order to obtain the transition probability. This is important, as on near-term hardware the exact implementation of $e^{-i\tau \dot A}$ will often not be possible. Instead, Suzuki-Trotter decompositions may be used. The approximate exponential is then $\tilde U(\tau) \approx \exp(-i\tau \dot A) = U(\tau)$. The key is that an $n$th-order Suzuki-Trotter decomposition yields an error of $\|U(\tau)-\tilde U(\tau)\|= \m O(\tau^{n+1})$ \cite{suzuki1985decomposition}, implying that a simple 1st-order decomposition yields $\tilde U(\tau) = 1 - i\tau \dot A + \m O(\tau^2)$. 

Under such a first-order approximation, the resulting extrapolation function takes the form
\begin{equation}
\begin{split}
\tilde f(\tau) = |\la \ac a| \tilde U(\tau) |\g b \ra|^2 + |\la \ac a| \tilde U(-\tau) |\g b \ra|^2 \\= 2\tau^2 |\la a|A|b \ra|^2 + \tau^4 \tilde K_4 + \mO(\tau^6) + \cdots
\end{split}
\end{equation}
where the leading term is the same as before. 
Fortuitously, all of this implies that there is no immediately obvious reason to think an exact implementation of $\exp(-i\tau A)$ is more useful than a 1st-order Suzuki-Trotter approximation.

\subsection{Shot count analysis for NOTraP-HD}\label{sec:shotstheory}

Ultimately it is the number of circuit shots (as opposed to number of unique circuits) and the circuit depth that are most relevant in a resource analysis. In this section we analyze the number of shots (i.e. circuit repetitions) required when implementing NOTraP-HD. Here we consider only the case of
$n_\tau=2$ extrapolation points; the analysis would proceed analogously when considering $n_\tau>2$.

It is useful to express the Richardson extrapolation from equation \eqref{eq:f_expan} in the form \cite{pozrikidis2008} 

\begin{equation}
\begin{split}
T \vec c &= \vec f  \\
\begin{bmatrix}
\tau_0^2 & \tau_0^4  \\
\tau_1^2 & \tau_1^4  \\
\end{bmatrix}
\begin{bmatrix}
K_2 \\ K_4 \\
\end{bmatrix} &=
\begin{bmatrix}
f_0  \\ f_1  \\
\end{bmatrix}
\end{split}
\end{equation}
where $f(\tau_i)=f_i$ and the goal of the extrapolation is to find $K_2/2=Q' \approx \aAb$. We define $V=T^{-1}$, for which the analytical expression is
\begin{equation}
\begin{pmatrix}
\tau_0^2 & \tau_0^4 \\
\tau_1^2 & \tau_1^4
\end{pmatrix}
^{-1}
= V =
\frac{1}{\tau_0^2-\tau_1^2}
\begin{pmatrix}
-\tau_1^2/\tau_0^2 & \tau_0^2/\tau_1^2 \\
1/\tau_0^2 & -1/\tau_1^2
\end{pmatrix}
\end{equation}

Notably, each $f(\tau)$ value requires two distinct circuit evaluations, each of which will be associated with an uncertainty due to a finite number of circuit shots. Decomposing $f(\tau)$ by defining
\begin{equation}
\begin{split}
f(\tau_i) &= f_i^+ + f_i^- \\
f_i^- &= |\la \ac a | e^{-i\tau_i \dot A} | \g b \ra|^2 \\
f_i^+ &= |\la \ac a | e^{+i\tau_i \dot A} | \g b \ra|^2
\end{split}
\end{equation}
leads to
\begin{equation}\label{eq:Qestimator}
Q' = \half K_2 = \half \left [  V_{00} f_0^+ + V_{00} f_0^- + V_{01} f_1^+ + V_{01} f_1^-  \right ] .
\end{equation}
Hence (for $n_\tau=2$) we have four circuit evaluations in determining the transition probability. %

There are two types of error in NOTraP-HD. First is measurement uncertainty $\varepsilon_{meas}$, the result of a finite number of measurements when evaluating $f(\tau)$. The second is extrapolation error $\varepsilon_{extrap}$,
\begin{equation}
\varepsilon_{extrap} = Q - Q' = \aAb - \half K_2,
\end{equation}
resulting from the Taylor approximation. $\varepsilon_{extrap}$ can be made arbitrarily small by decreasing $\tau$.

The measurement uncertainty derived from equation \eqref{eq:Qestimator} is
\begin{equation}\label{eq:Qmeasuncert}
\varepsilon_{meas}^2 = \frac{1}{4} \left [ 
V_{00}^2 \varepsilon_{f_0^+}^2 + 
V_{00}^2 \varepsilon_{f_0^-}^2 + 
V_{01}^2 \varepsilon_{f_1^+}^2 + 
V_{01}^2 \varepsilon_{f_1^-}^2 \right ].
\end{equation}

The expression for the measurement uncertainty $\varepsilon_f$ in each $f_i^{\pm}$ is 
\begin{equation}
\varepsilon_{f_i^{\pm}}^2 = \text{Var}[f_i^{\pm}]/N_{f_i^{\pm}}
\end{equation}
where $N_{f_i^{\pm}}$ is the number of circuit shots for $f_i^{\pm}$. The number of required shots for a circuit $f_i^{\pm}$ will tend to be higher when its coefficient in equation \eqref{eq:Qestimator} is larger. Hence when determining $N_{tot}=\sum N_{f_i^{\pm}}$ in the results of Section \ref{sec:sims} we use $N_{f_i^{\pm}} \propto V_{0i}^2$, which is analogous to the strategy used in the context of calculating expectation values \cite{crawford2021si,yen2023deterministic}. 

Finally we discuss the choice of $\tau$ values with the goal of reducing circuit shots, where $\tau_1>\tau_0$. The choice of $\tau_1$ is arbitrary, but it should be as large as possible in order to reduce required shots, while being small enough to ensure that extrapolation error is sufficiently small. Now, because equation \eqref{eq:Qmeasuncert} implies that larger $V_{0i}$ leads to more required shots, we would like $V_{0i}$ as small as possible. As $V_{00}>V_{01}$, a reasonable strategy is to choose $\tau_0$ such that $V_{00}=\frac{\tau_1^2}{\tau_0^2(\tau_1^2-\tau_0^2)}$ is minimized. This minimum occurs at $\tau_0=\tau_1/\sqrt{2}$, which is the value we use for our numerical simulations in Section \ref{sec:results}.

\subsection{Tuning between low depth and few circuit evaluations (NOTraP-T)}\label{sec:interpmethod}

\begin{figure}
    \centering
    \includegraphics[width=.5\textwidth]{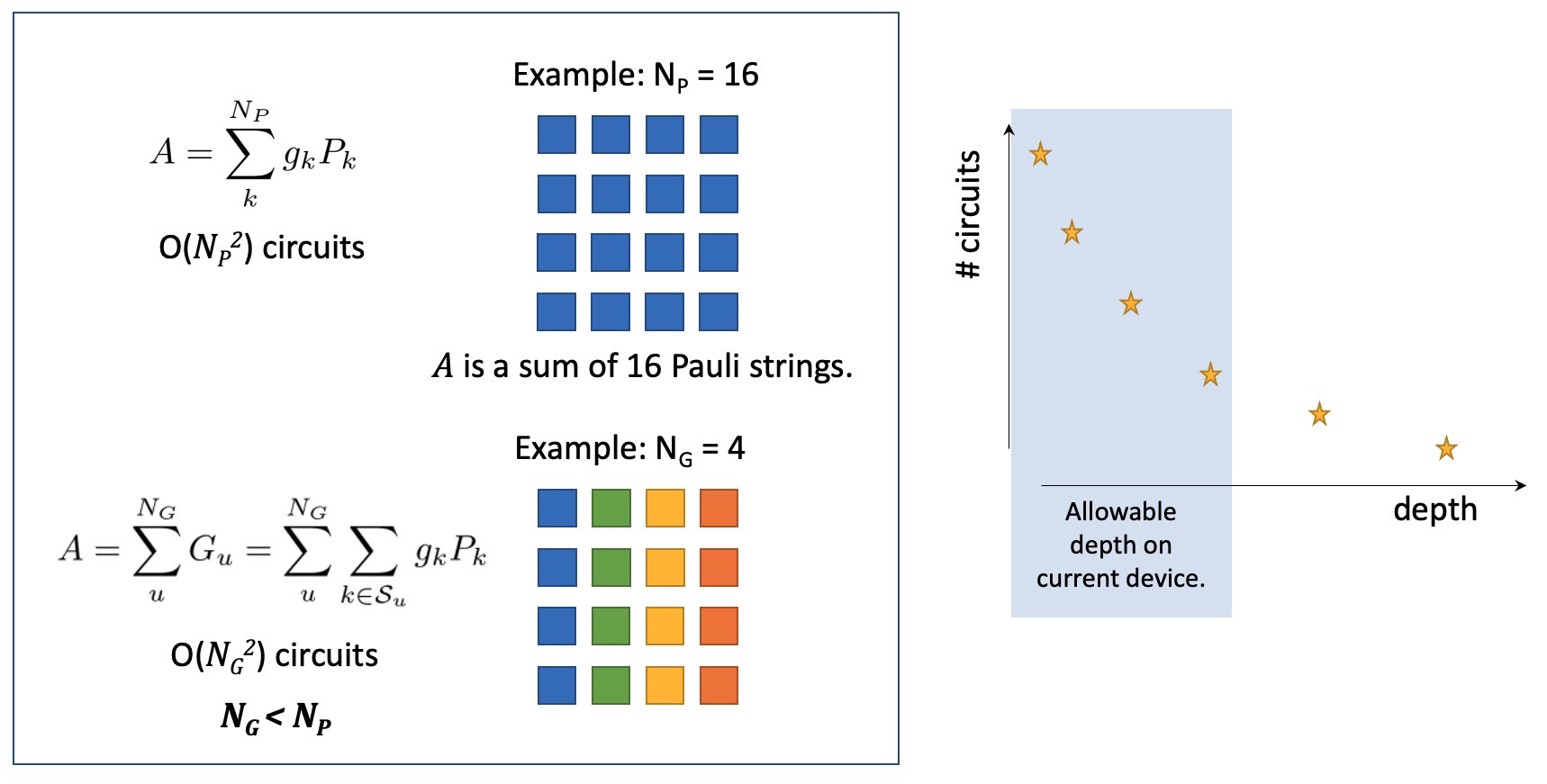}
    \caption{Schematic of the tunable NOTraP-T algorithm.}
    \label{fig:concept}
\end{figure}

Different near- and medium-term quantum hardware will have different characteristics---and whatever hardware one uses, one generally would like to take advantage of as much of the circuit depth as possible. This suggests the need for algorithms whereby the circuit depth may be tuned over a wide range, at the cost of another resource like (in this case) total circuit evaluations (see Figure \ref{fig:concept}).

One of the above methods (\notrap-SD) uses shorter-depth circuits to implement transition probabilities, while the other (\notrap-HD) requires deeper circuits but fewer measurements. In this section we introduce a method by which one may interpolate between the short-depth NOTraP-SD and the high-depth NOTraP-HD, which allows for a more fine-tuned trade-off between depth and total circuit evaluations.

Expressing $A$ as a linear combination of $N_P$ Pauli strings, the idea is to first group $A$ into $N_G$ terms instead, with $N_G<N_P$. Write down
\begin{equation}\label{eq:interp_decomp}
A = \sum_u^{N_G} G_u = \sum_u^{N_G} \sum_{k \in \mS_u} g_k P_k 
\end{equation}
where each $G_u$ is a linear combination of Pauli strings, $\mS_u$ $\subset$ $\{0,\cdots,N_P-1\}$, $\bigcup_u \mS_u = \{0,\cdots,N_P-1\}$, and $\sum_u|\mS_u|=N_P$.

One may choose any arbitrary decomposition for equation \eqref{eq:interp_decomp}; there is no requirement, for example, that members of the set commute. It may appear as though one would want the members of every $\{g_k | k \in \mS_u\}$ to commute, as that would allow one to exactly implement $e^{-i \theta G_u}$ via first-order Trotterization. However, because it is only fourth-order ($\tau^4$) terms and above that are affected by first-order Trotter error (see Section \ref{sec:expmethod}) in equation \eqref{eq:f_expan}, there is no immediately obvious advantage to having Pauli strings commute. The total number of distinct circuits to run and measure is then reduced to $\mO(N_G^2)$, from the previous $\mO(N_P^2)$.

Though distinct, our route to the derivation is somewhat inspired by formulas \eqref{eq:w1234} and \eqref{eq:aAb_sq_1}. Define the terms
\begin{equation}
\begin{split}
S_{+}\ep{u} &= |\la \ac a|e^{-i\tau \dot G_u}|\g b \ra|^2, \\
\end{split}
\end{equation}
\begin{equation}
\begin{split}
S_{-}\ep{u} &= |\la \ac a|e^{+i\tau \dot G_u}|\g b \ra|^2, \\
\end{split}
\end{equation}

\begin{equation}
\begin{split}
S_-\ep{uv} &= |\la \ac a|e^{-i\tau \dot G_u}e^{-i\tau \dot G_v}|\g b \ra|^2 \\
&= |\la \ac a|[1 - i\tau \dot G_u + \mO(\tau^2)][1 - i\tau \dot G_v + \mO(\tau^2)]|\g b \ra|^2 ,
\end{split}
\end{equation}
\begin{equation}
\begin{split}S_+\ep{uv} &= |\la \ac a|e^{+i\tau \dot G_u}e^{+i\tau \dot G_v}|\g b \ra|^2 \\
&= |\la \ac a|[1 + i\tau \dot G_u + \mO(\tau^2)][1 + i\tau \dot G_v + \mO(\tau^2)]|\g b \ra|^2. \\
\end{split}
\end{equation}

For a single $G_u$,
\begin{equation}
\begin{split}
S_{+}\ep{u} + S_{-}\ep{u} = 2\tau^2 |\la a|G_u|b \ra|^2 + \mO(\tau^4).  \\
\end{split}
\end{equation}

It is instructive to first consider the case of $N_G=2$.
We write 
\begin{equation}
\begin{split}
S_+\ep{uv} + S_-\ep{uv}  \\
&= 4\tau^2 \Re \la a|G_u|b \ra\la b|G_v|a \ra \\
&+ 2\tau^2 |\la a|G_u|b \ra|^2 \\
&+ 2\tau^2 |\la a|G_v|b \ra|^2 + \mO(\tau^4)
\end{split}
\end{equation}
which upon rearrangement (when $N_G=2$) yields 
$|\la a|A|b \ra|^2 = \frac{1}{2\tau^2}(S_+\ep{uv} + S_-\ep{uv}) + \mO(\tau^4)$. 

Turning to the general case of $N_G \geq 2$, it can be shown that
\begin{equation}
\begin{split}
|\la a|A|b \ra|^2 = \\
\sum_u^{N_G} (S_+\ep{u}+S_-\ep{u})/2\tau^2 + 
\sum_u^{N_G} \sum_{v < u} \big[ S_+\ep{uv} + S_-\ep{uv} \\
- S_+\ep{u} - S_+\ep{v} - S_-\ep{u} - S_-\ep{v} + \big]/2\tau^2 + \mO(\tau^4)/\tau^2  \\
\end{split}
\end{equation}

or

\begin{equation}\label{eq:aAb_interp}
\begin{split}
|\la a|A|b \ra|^2 = |\la a|\sum_u G_u|b \ra|^2 \\
= \sum_u^{N_G} \sum_{u<v} \big( S_+\ep{uv} + S_-\ep{uv} \big)/2\tau^2 \\-  \sum_{u}\big( N_G - 2 \big)(S_+\ep{u}+S_-\ep{u})/2\tau^2 \\
+ \mO(\tau^2).
\end{split}
\end{equation}

This method (NOTraP-T) is tunable via changing the number of subsets $N_G$. As before, the orthogonalization procedure of Lemmas \ref{lem:1} and \ref{lem:2} is strictly necessary for the above expressions to yield the correct result. Note that the main mathematical way in which NOTraP-T differs from NOTraP-SD is that the former formulas do not explicitly use coefficients of the Pauli matrices. This difference is what leads to the $\frac{1}{2\tau^2}$ factor.

Finally, the algorithm proceeds by evaluating formula \eqref{eq:aAb_interp} for multiple values of $\tau$ and then extrapolating. %

There are $N_G$ circuits each for types $S\ep{u}_+$ and $S\ep{u}_-$, and $\half(N_G^2-N_G)$ circuits each for types $S\ep{uv}_+$ and $S\ep{uv}_-$. Hence the number of circuits required to estimate the transition probability  is $n_{\tau}(N_G^2+N_G)$ where $n_{\tau}$ is the number of points used in the extrapolations. As mentioned in the previous section, the relationship between choice of $\tau$ and the overall total number of measurements is nontrivial, and we leave a detailed analysis of this relationship in the tunable NOTraP-T to future work.

In this section we have provided an algorithm with tunable circuit depth, where the availability of more circuit depth allows one to evaluate expectation values from fewer circuits. The tunability comes in the choice of the number of groups $N_G$ in which to place the Pauli terms of operator $A$. This approach is conceptually a mix between the low-depth method of Section \ref{sec:nonorthog} and the high-depth method of Section \ref{sec:expmethod}, allowing the user to make use of the circuit depth available on a particular quantum device.

\section{Applications and numerical simulations}\label{sec:sims}

Here we implement numerical proofs-of-principle for calculating transition probabilities via the NOTraP algorithms. Our primary focus is on the high-depth NOTraP-HD method, whose viability we wish to demonstrate using very few extrapolation points. Additionally, we study the tradeoffs between circuit depth and the number of circuit evaluations for all algorithms introduced in this work. %
We consider simple models from three application areas: spin chains, vibronic transitions in molecules, and linear systems solving. We prepared and executed the numerics using primarily Scipy \cite{scipy}, mat2qubit \cite{mat2qubit,sawaya2022dqir}, and OpenFermion \cite{openfermion}. In Section \ref{sec:methods} we summarize the applications and numerical methods and in Section \ref{sec:results} we present and discuss the results.

\subsection{Applications and numerical methods}\label{sec:methods}

\begin{figure}
    \centering
    \includegraphics[width=.5\textwidth]{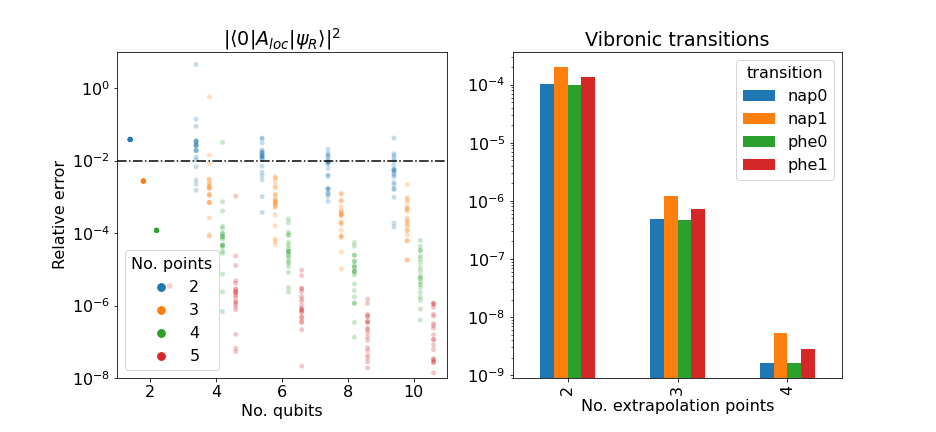}
    \caption{Errors via the exponentiation- and extrapolation-based method NOTraP-HD using $n_\tau$ values of 2 through 5, for $A_{loc}$ and for selected vibronic transitions in napthalene and phenanthrene. ($A_{nonloc}$ yields virtually identical results as $A_{loc}$.) The labels nap\{$N$\} and phe\{$N$\} refer to transitions from the ground state of napthalene and phenanthrene, respectively, to the $N$th vibrational excited state of the electronic excited state. %
    }
    \label{fig:extrap-Aloc-vibron}
\end{figure}

\begin{figure}
    \centering
    \includegraphics[width=.5\textwidth]{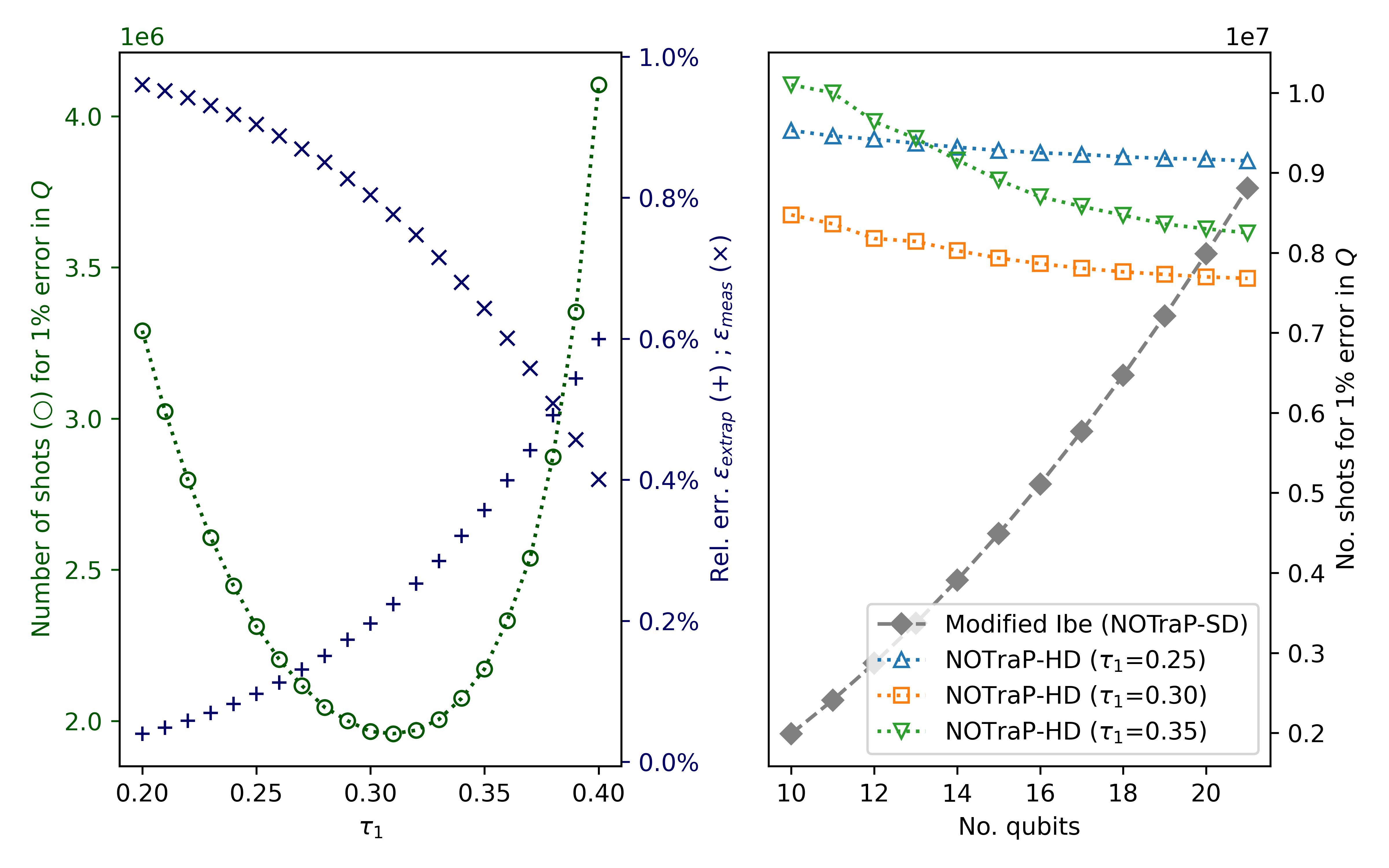}
    \caption{Number of measurement shots (circuit repetitions) required in modified Ibe (NOTraP-SD) and our extrapolation method (NOTrap-HD) with respect to operator $A_{nonloc}$, using $n_\tau=2$. \textbf{Left}: Number of shots required with respect to $\tau_1$, for $n_q=16$ qubits. These are the shot counts (green circles) required to achieve $\varepsilon_{extrap} + \varepsilon_{meas}<1\%$, where $\varepsilon_{extrap}$ is the extrapolation error (blue `+') and $\varepsilon_{meas}$ is the measurement uncertainty (blue `$\times$'). \textbf{Right}: Number of shots required with respect to qubit count. NOTraP-HD requires fewer shots than modified Ibe (NOTraP-SD) after approximately 20 qubits.
    }
    \label{fig:shots}
\end{figure}

\begin{figure}
    \centering
    \includegraphics[width=.5\textwidth]{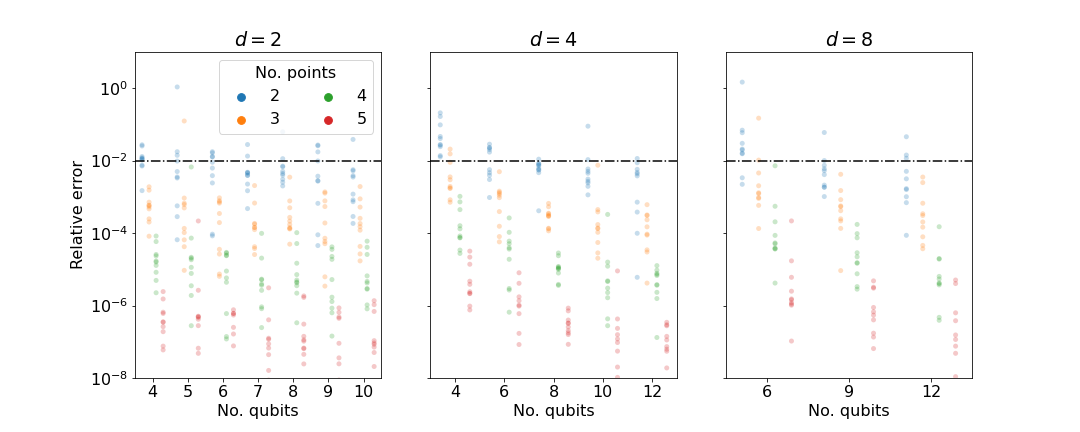}
    \caption{Using NOTraP-HD of Section \ref{sec:expmethod} to estimate $|\la x|A|b \ra|^2$ for 20 random instances using $A$ as defined in equation \eqref{eq:ttrain}, relevant to classical linear algebra. The number of extrapolation points $n_\tau$ is varied from 2 to 5.}
    \label{fig:extrap-linalg}
\end{figure}

\textbf{Toy spin model.} 
We implement both the one-local operator $A_{loc}$ and the ($N_q-1$)-local operator $A_{nonloc}$ defined in Section \ref{sec:theory}. While $A_{nonloc}$ is a contrived operator introduced mainly to study a highly non-local toy case, $A_{loc}$ has a straight-forward physical interpretation as the total magnetism in the transverse direction. For up to ten qubits, we use the extrapolated algorithm NOTraP-HD to calculate $|\la 0^{\otimes N_q} |A_{loc}| \psi_R \ra|^2$ where $|\psi_R \ra$ is a random quantum state. The extrapolation points $\tau$ were arbitrarily chosen to be centered arond $\frac{1}{\|A\|}$ with spacing between points of $\frac{0.1}{\|A\|}$, where $\|A\|$ is the spectral norm. 
In analyzing NOTraP-HD, we analyze the circuit depth versus distinct quantum circuits required for both $A_{loc}$ and $A_{nonloc}$. As before, we assume a native gate set of CNOT and all arbitrary one-qubit gates. For \notrap-T, we group the Pauli terms into $N_G$ equally sized sets.

\textbf{Vibronic transitions in molecules. }
In the general case, vibrational degrees of freedom in molecules are hard to simulate on classical computers; there has been theoretical work towards developing general quantum approaches for treating vibrations \cite{mcardle19_qvibr,ollitrault20_reiher_qvibr,sawaya2021ir}. The more specific problem of calculating vibronic spectra involves the calculation of light-absorption cross sections for coupled vibrational-electronic transitions \cite{huh15,huh17_manhong,sawaya19_vibronic,jahangiri2020xanaduvibronic}; here, the goal is to determine both the frequencies and the absorption intensities of each transition. The transition occurs between two electronic eigenstates, each associated with a different vibrational Hamiltonian. Though highly accurate vibrational Hamiltonians must include anharmonic terms, here we approximate napthalene and phenanthrene using harmonic potential energy surfaces. The transformation between the ground and excited vibrational normal mode coordinates is described by the Duschinsky transformation,
\begin{equation}
\vec{q'} = \mathbf S \vec q + \vec d
\end{equation}
where $q'$ and $q'$ are vibrational normal mode coordinates respectively for the excited and ground electronic potential energy surface (PES), $\mathbf S$ is a unitary matrix, and $d$ is a displacement vector. The transition intensity is described by $|\la \psi|\mu|\psi' \ra|^2$, where $|\psi\ra$ is a vibrational eigenstate of the ground PES and $|\psi'\ra$ is a vibrational eigenstate of the excited PES. When the so-called Condon approximation is used, $\mu=I$ and one simply calculates the overlap of the two states. Non-Condon effects must be included in order to qualitatively capture the correct spectrum in many molecules, including those considered here. An analog photonic-based algorithm for the inclusion of non-Condon effects has been previously described \cite{jnane2021noncondon}; the current work can be used as a digitial version.

We use the following parameters\cite{jnane2021noncondon} for napthalene,
\begin{equation}
\begin{split}
S\supr{\rm{nap}} = \begin{pmatrix}
.98 & -.20 \\
.20 & .98
\end{pmatrix}
\\
\omega_g\supr{\rm{nap}} = (509,938)
\\
\omega_e\supr{\rm{nap}} = (438,912)
\\
d\supr{\rm{nap}} = (0,0)
\\
\mu\supr{\rm{nap}} = I + q_0 - q_1,
\end{split}
\end{equation}
and the following for phenanthrene,
\begin{equation}
\begin{split}
S\supr{\rm{phen}} = \begin{pmatrix}
.9055 & -.4240 \\
.4240 & .9055
\end{pmatrix}
\\
\omega_g\supr{\rm{phen}} = (700,800)
\\
\omega_e\supr{\rm{phen}} = (679,796)
\\
d\supr{\rm{phen}} = (.1650,.0780)
\\
\mu\supr{\rm{phen}} = I + 1.5 q_0 - 0.5 q_1.
\end{split}
\end{equation}

In both problems, we consider only the two normal modes most relevant to the spectra. We map the problem to a binary representation \cite{sawaya20_dlev,sawaya2020connectivity} using a truncation of 16 energy levels, leading to 4 qubits per mode and hence 8 qubits overall. We deliberately chose molecules for which vibrational eigenstates on different PESs are not orthogonal, in order to highlight the utility of NOTraP.

\textbf{Linear systems.} 
A ubiquitous problem in science and engineering is solving the linear system
\begin{equation}\label{eq:Ax_b}
A|x\ra = |b\ra
\end{equation}
where $A$ and $|b\ra$ are known (or in the quantum case, it is at least known how to prepare $|b\ra$). The first quantum algorithm for solving linear systems \cite{harrow2009_hhl} uses the quantum phase estimation algorithm. Variational approaches to solve \eqref{eq:Ax_b} have been developed \cite{bravo19_vqls,xu19_vari_linalg,endo20_vari_genproc,wang2021vqsvd,huang2021near}, based on the notion that one may vary $|x\ra$ using the objective function %
\begin{equation}\label{eq:vqls_obj}
F_{LS} = \la x|A|b \ra .
\end{equation}
or $\|F_{LS}\|^2$, where the correct answer is reached when this quantity is maximized. Note that the original VQLS paper uses a slightly modified definition of the cost function that does not qualitatively affect the results of the current work.

Equation \ref{eq:vqls_obj} is simply a transition amplitude, and so any method for calculating either transition amplitudes or transition probabilities may be used. The methods of the current work may be used to solve for $|x\ra$, using $|\la x|A|b \ra|^2$ as an objective function instead of equation \eqref{eq:vqls_obj}. %
Note that the original \cite{ibe20_transition} formulas \eqref{eq:aAb_sq_1}  would not work for optimizing the objective function because they require $|b\ra$ to be orthogonal to the candidate for $|x\ra$.

In order to study the use of our subroutines in variational quantum linear systems (VQLS) solving, we consider random classical matrices that can be expressed as \textit{tensor trains} \cite{bigoni2016tensortrain,novikov2020tensortrain}. Commonly used in machine learning and other fields, tensor train decompositions are a linear combination of matrix tensor products. We use the simple structure
\begin{equation}\label{eq:ttrain}
A = \sum_{i=1}^{N_{\text{train}} \text{-} 1} R^{(i)}_{i} \otimes R^{(i')}_{i+1}
\end{equation}
where the size of each local subspace is $d$, hence each ``local'' matrix $R$ is of size $d \times d$. The subscripts denote the position in the tensor network train (identities $I$ are implicit on other subspaces), and we include superscripts to clarify that all $R$ are unique. We consider local tensor sizes $d=2,4,8$, corresponding to local qubit counts of 1, 2, and 3, respectively. The number of qubits required to represent each full classical matrix $A$ is thus $N_{train} \times  \log_2 d$. We consider the error in calculating $|\la b|A|x \ra|^2$, where $|b\ra$ is the zero vector, for 20 random vectors $|x\ra$. Though we constrain $A$ to be Hermitian, it is known that any non-Hermitian matrix can be made Hermitian using one additional qubit, as shown in Appendix \ref{apx:nonherm}.

\subsection{Results and Discussion}\label{sec:results}

\begin{figure}
    \centering
    \includegraphics[width=\linewidth]{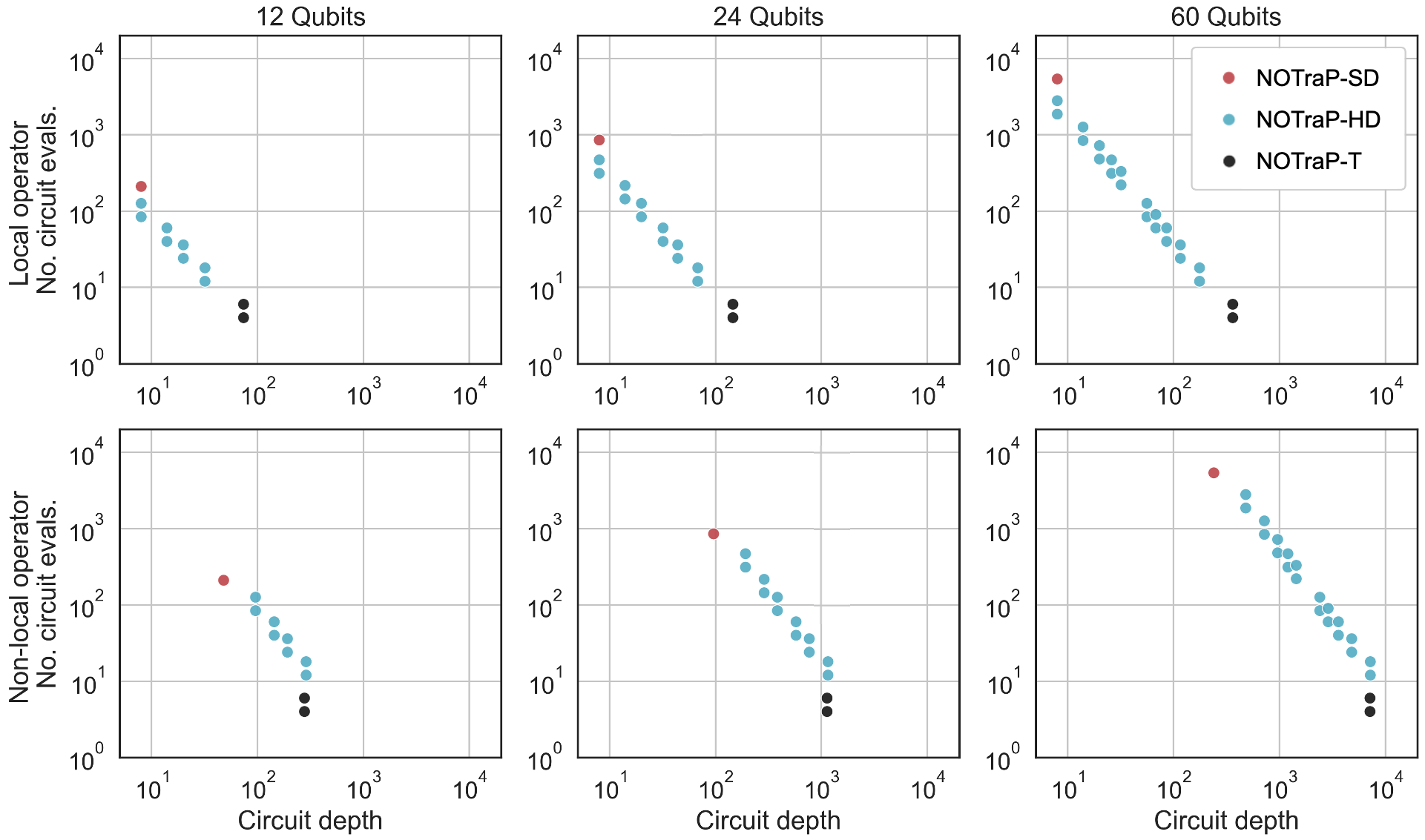}
    \caption{Pareto fronts demonstrating the trade-off between circuit depth and distinct quantum circuits to evaluate, for the short-depth \notrap-SD, the tunable NOTraP-T, the high-depth \notrap-HD. Circuit depths from the method of reference \cite{bravo19_vqls} would be similar to \notrap-SD results. Two operators, the local $A_{loc}$ and the non-local $A_{nonloc}$, are considered. (Reported circuit depths do not include state preparation unitaries, whose depths vary greatly depending on the problem instance. The reason we do not compare against the Hadamard test is that the method would require controlled state preparation unitaries, leading to an overall decomposed circuit depth much larger than NOTraP. The Hadamard test is virtually guaranteed to produce circuits much deeper than these)}
    \label{fig:pareto}
\end{figure}

The left panel of Figure \ref{fig:extrap-Aloc-vibron} shows the relative error of using \notrap-HD for the spin operator $A_{loc}$, for several qubits counts and values of $n_\tau$. The right panel of Figure \ref{fig:extrap-Aloc-vibron} shows vibronic transitions discussed in the previous section. In both applications, we observe a trend of the error decreasing exponentially with the number of extrapolation points used (for a similar range of $\tau$), demonstrating the viability of \notrap-HD for even very few extrapolation points. It is also notable that just 2 extrapolation points is often sufficient for a sub-$1\%$ error for these choices of $\tau$, while 3 or more extrapolation points is almost always sufficient. In the case of $A_{loc}$, the average error scaling shows an improvement as the number of qubits increases; we do not assume this is a general result though.

Next we studied the number of shots required to yield 1\% error, using the operator $A_{nonloc}$. Figure \ref{fig:shots} shows shot counts for $\| \langle a | A_{nonloc} | b \rangle\|^2$, with states chosen such that $\| \langle a | A_{nonloc} | b \rangle\|^2 = 0.5$, where we set $\tau_0=\tau_1/\sqrt{2}$ (see Section \ref{sec:shotstheory}). The left panel shows the number of shots required to achieve measurement error $\varepsilon_{meas}$ such that $\varepsilon_{extrap} + \varepsilon_{meas}=1\%$. Importantly, as $\tau_1$ increases so does $\varepsilon_{extrap}$. In turn, this leads to a smaller allowable $\varepsilon_{meas}$ and hence more required shots as $\tau_1$ increases. Conversely, a smaller $\tau_1$ leads to larger values of $V_{0i}$, which leads to more shots. This trade-off leads to a ``sweet spot'' for $\tau_1$ that is evident in the left panel, where the minimum number of shots is achieved at $\tau_1 \approx 0.31$.

The right panel of Figure \ref{fig:shots} shows modified Ibe (NOTrap-SD) as well as NOTrap-HD with $\tau_1 \in \{0.25,0.30,0.35\}$, for 10 through 21 qubits. For modified Ibe we set $\varepsilon_{meas}=1\%$ since unlike NOTraP-HD there is no extrapolation error; as before, for NOTraP-HD the number of shots ensures that $\varepsilon_{meas}$ = $1\%-\varepsilon_{extrap}$. After approximately 20 qubits, NOTraP-HD requires fewer total shots than modified Ibe (NOTraP-SD). This shows a clear example where the extrapolation algorithm out-performs modified Ibe in terms of shot counts.

The results show that unlike modified Ibe, the number of shots for NOTraP-HD is not directly related to the number of Pauli terms $N_P$. This is unsurprising, because in some sense the Hamiltonian is taken into account ``all at once'' in $f(\tau_i)$, whereas in Ibe's method all terms and cross-terms require separate circuits. In these shot count analyses we have calculated the extrapolation error explicitly---in future work it should be possible to developing analytical approaches to upper bound this error, as in a real implementation one would not a priori know the extrapolation error. %

Figure \ref{fig:extrap-linalg} shows relative error in calculating $|\la b|A|x \ra|^2$ for a classical $A$ defined by formula \eqref{eq:ttrain}. We observe similar trends as before: errors decrease exponentially with increasing $n_\tau$ for a roughly constant range of $\tau$, and there is a modest trend of error reduction as the qubit count increases. The value of the local tensor size $d$ does not appear to have a strong effect on error.

Figure \ref{fig:pareto} shows Pareto fronts with respect to circuit depth and number of circuit evaluations for several methods, with the purpose of comparing resource counts and especially demonstrating the resource trade-offs for the tunable \notrap-T method. We consider both $A_{loc}$ and $A_{nonloc}$; the latter leads to larger circuit depths as the operator is highly non-local. \notrap-T and \notrap-HD show two values per circuit depth, as we include $n_\tau$ values of both 2 and 3. We stress again that these plots show circuit depths only for the calculation of $\aAb$ (depths of the required state preparation unitaries are excluded) and that the relationship between number of circuit evaluations and total number of measurements is non-trivial.

The plots show that \notrap-T may be continuously tuned between high circuit depth and high number of circuit evaluations. As previously stated, this allows for the \notrap~methods to be tailored to the particular limited depth requirements of a particular quantum device. 
These circuit depths are based on the gate decompositions discussed in Section \ref{sec:prevmethods}, and highlight the fact that \notrap~allows for shorter depth circuits than previous methods. Note that once the available circuit depth is larger than the depth required for \notrap-HD, there is no benefit to using \notrap-T.

\section{Conclusions}

In this work we made three contributions to quantum algorithms for calculating transition probabilities. First, building on important previous work \cite{ibe20_transition} we introduced a method (\notrap-SD) that can be used to calculate transition probabilities between two non-orthogonal states, while still (like another recently proposed method \cite{bravo19_vqls}) not requiring any controlled state preparations nor controlled Pauli rotations. Second, we introduced an extrapolation-based method (\notrap-HD) that needs enough circuit depth to implement a first-order Suzuki-Trotter circuit but requires asymptotically far fewer measurements. Third, we show how to tune (\notrap-T) between these short- and high-depth methods in order to make the most of available quantum hardware. %
Note that in principle this method might be used in other advanced methods like shadow tomography \cite{huang2020shadow}, which for highly local Hamiltonians may show large reductions in measurements required per circuit evaluation. %

There are several other promising directions for future work. One can envision using methods here for calculating linear combinations of higher-order operators $\la A^K \ra$ \cite{seki21_power}. Further, there may be modifications to our methods that allow for fewer measurement counts when calculating expectation values \cite{mcclean16_njp,verteletskyi2020measurement,cerezo2021variational,fedorov2022vqerev} as opposed to transition probabilities. 

In conclusion, we expect these quantum subroutines to be useful on near- and mid-term quantum computers for a range of problems in chemistry, materials, condensed matter physics, and quantum linear algebra, especially when measurement shots and/or circuit depth are the limiting resources.

\section*{Acknowledgements}

We are grateful for useful correspondence with Albert Schmitz and Ra\'ul Garc\'ia-Patr\'on. 
JH acknowledges the support of Basic Science Research Program through the National Research Foundation of Korea (NRF), funded by the Ministry of Education, Science and Technology (NRF-2021M3E4A1038308, NRF-2021M3H3A1038085, NRF-2019M3E4A1079666, NRF-2022M3H3A106307411, NRF-2023M3K5A1094805, 2023M3K5A1094813).

\appendix

\section{non-Hermitian matrices}\label{apx:nonherm}

If $A$ is not Hermitian, for example if it was lifted from a classical matrix problem instance, it can be made Hermitian by 
\begin{equation}
A \rightarrow \begin{bmatrix}
0 & A \\
A^\dag & 0 \\
\end{bmatrix}
\end{equation}
which is done by first decomposing the operator into Hermitian and anti-Hermitian parts,
\begin{equation}
A = A_H + A_{AH}
\end{equation}

where $A_H=(A+A^\dag)/2$ and $A_{AH} = (A-A^\dag)/2$. Adding one qubit and the following transformation yields a Hermitian operator:

\begin{equation}
A \rightarrow X \otimes A_H + Y \otimes (-i) A_{AH}.
\end{equation}

State vectors are modified as
\begin{equation}
\begin{split}
| b\ra &\rightarrow |0\ra \otimes |b\ra \\
| x\ra &\rightarrow |1\ra \otimes |x\ra.
\end{split}
\end{equation}

\bibliography{main}%

\providecommand{\noopsort}[1]{}\providecommand{\singleletter}[1]{#1}%
\begin{thebibliography}{73}%
\makeatletter
\providecommand \@ifxundefined [1]{%
 \@ifx{#1\undefined}
}%
\providecommand \@ifnum [1]{%
 \ifnum #1\expandafter \@firstoftwo
 \else \expandafter \@secondoftwo
 \fi
}%
\providecommand \@ifx [1]{%
 \ifx #1\expandafter \@firstoftwo
 \else \expandafter \@secondoftwo
 \fi
}%
\providecommand \natexlab [1]{#1}%
\providecommand \enquote  [1]{``#1''}%
\providecommand \bibnamefont  [1]{#1}%
\providecommand \bibfnamefont [1]{#1}%
\providecommand \citenamefont [1]{#1}%
\providecommand \href@noop [0]{\@secondoftwo}%
\providecommand \href [0]{\begingroup \@sanitize@url \@href}%
\providecommand \@href[1]{\@@startlink{#1}\@@href}%
\providecommand \@@href[1]{\endgroup#1\@@endlink}%
\providecommand \@sanitize@url [0]{\catcode `\\12\catcode `\$12\catcode
  `\&12\catcode `\#12\catcode `\^12\catcode `\_12\catcode `\%12\relax}%
\providecommand \@@startlink[1]{}%
\providecommand \@@endlink[0]{}%
\providecommand \url  [0]{\begingroup\@sanitize@url \@url }%
\providecommand \@url [1]{\endgroup\@href {#1}{\urlprefix }}%
\providecommand \urlprefix  [0]{URL }%
\providecommand \Eprint [0]{\href }%
\providecommand \doibase [0]{https://doi.org/}%
\providecommand \selectlanguage [0]{\@gobble}%
\providecommand \bibinfo  [0]{\@secondoftwo}%
\providecommand \bibfield  [0]{\@secondoftwo}%
\providecommand \translation [1]{[#1]}%
\providecommand \BibitemOpen [0]{}%
\providecommand \bibitemStop [0]{}%
\providecommand \bibitemNoStop [0]{.\EOS\space}%
\providecommand \EOS [0]{\spacefactor3000\relax}%
\providecommand \BibitemShut  [1]{\csname bibitem#1\endcsname}%
\let\auto@bib@innerbib\@empty
\bibitem [{\citenamefont {Wecker}\ \emph {et~al.}(2015)\citenamefont {Wecker},
  \citenamefont {Hastings}, \citenamefont {Wiebe}, \citenamefont {Clark},
  \citenamefont {Nayak},\ and\ \citenamefont {Troyer}}]{wecker15}%
  \BibitemOpen
  \bibfield  {author} {\bibinfo {author} {\bibfnamefont {D.}~\bibnamefont
  {Wecker}}, \bibinfo {author} {\bibfnamefont {M.~B.}\ \bibnamefont
  {Hastings}}, \bibinfo {author} {\bibfnamefont {N.}~\bibnamefont {Wiebe}},
  \bibinfo {author} {\bibfnamefont {B.~K.}\ \bibnamefont {Clark}}, \bibinfo
  {author} {\bibfnamefont {C.}~\bibnamefont {Nayak}},\ and\ \bibinfo {author}
  {\bibfnamefont {M.}~\bibnamefont {Troyer}},\ }\bibfield  {title} {\bibinfo
  {title} {Solving strongly correlated electron models on a quantum computer},\
  }\href {https://doi.org/10.1103/PhysRevA.92.062318} {\bibfield  {journal}
  {\bibinfo  {journal} {Phys. Rev. A}\ }\textbf {\bibinfo {volume} {92}},\
  \bibinfo {pages} {062318} (\bibinfo {year} {2015})}\BibitemShut {NoStop}%
\bibitem [{\citenamefont {Dumitrescu}\ \emph {et~al.}(2018)\citenamefont
  {Dumitrescu}, \citenamefont {McCaskey}, \citenamefont {Hagen}, \citenamefont
  {Jansen}, \citenamefont {Morris}, \citenamefont {Papenbrock}, \citenamefont
  {Pooser}, \citenamefont {Dean},\ and\ \citenamefont
  {Lougovski}}]{dumitrescu18}%
  \BibitemOpen
  \bibfield  {author} {\bibinfo {author} {\bibfnamefont {E.~F.}\ \bibnamefont
  {Dumitrescu}}, \bibinfo {author} {\bibfnamefont {A.~J.}\ \bibnamefont
  {McCaskey}}, \bibinfo {author} {\bibfnamefont {G.}~\bibnamefont {Hagen}},
  \bibinfo {author} {\bibfnamefont {G.~R.}\ \bibnamefont {Jansen}}, \bibinfo
  {author} {\bibfnamefont {T.~D.}\ \bibnamefont {Morris}}, \bibinfo {author}
  {\bibfnamefont {T.}~\bibnamefont {Papenbrock}}, \bibinfo {author}
  {\bibfnamefont {R.~C.}\ \bibnamefont {Pooser}}, \bibinfo {author}
  {\bibfnamefont {D.~J.}\ \bibnamefont {Dean}},\ and\ \bibinfo {author}
  {\bibfnamefont {P.}~\bibnamefont {Lougovski}},\ }\bibfield  {title} {\bibinfo
  {title} {Cloud quantum computing of an atomic nucleus},\ }\href
  {https://doi.org/10.1103/PhysRevLett.120.210501} {\bibfield  {journal}
  {\bibinfo  {journal} {Phys. Rev. Lett.}\ }\textbf {\bibinfo {volume} {120}},\
  \bibinfo {pages} {210501} (\bibinfo {year} {2018})}\BibitemShut {NoStop}%
\bibitem [{\citenamefont {Bauer}\ \emph {et~al.}(2019)\citenamefont {Bauer},
  \citenamefont {de~Jong}, \citenamefont {Nachman},\ and\ \citenamefont
  {Provasoli}}]{bauer19_hep}%
  \BibitemOpen
  \bibfield  {author} {\bibinfo {author} {\bibfnamefont {C.~W.}\ \bibnamefont
  {Bauer}}, \bibinfo {author} {\bibfnamefont {W.~A.}\ \bibnamefont {de~Jong}},
  \bibinfo {author} {\bibfnamefont {B.}~\bibnamefont {Nachman}},\ and\ \bibinfo
  {author} {\bibfnamefont {D.}~\bibnamefont {Provasoli}},\ }\href@noop {}
  {\bibinfo {title} {A quantum algorithm for high energy physics simulations}}
  (\bibinfo {year} {2019}),\ \Eprint {https://arxiv.org/abs/arXiv:1904.03196}
  {arXiv:1904.03196} \BibitemShut {NoStop}%
\bibitem [{\citenamefont {Cao}\ \emph {et~al.}(2019)\citenamefont {Cao},
  \citenamefont {Romero}, \citenamefont {Olson}, \citenamefont {Degroote},
  \citenamefont {Johnson}, \citenamefont {Kieferov{\'{a}}}, \citenamefont
  {Kivlichan}, \citenamefont {Menke}, \citenamefont {Peropadre}, \citenamefont
  {Sawaya}, \citenamefont {Sim}, \citenamefont {Veis},\ and\ \citenamefont
  {Aspuru-Guzik}}]{cao19_rev}%
  \BibitemOpen
  \bibfield  {author} {\bibinfo {author} {\bibfnamefont {Y.}~\bibnamefont
  {Cao}}, \bibinfo {author} {\bibfnamefont {J.}~\bibnamefont {Romero}},
  \bibinfo {author} {\bibfnamefont {J.~P.}\ \bibnamefont {Olson}}, \bibinfo
  {author} {\bibfnamefont {M.}~\bibnamefont {Degroote}}, \bibinfo {author}
  {\bibfnamefont {P.~D.}\ \bibnamefont {Johnson}}, \bibinfo {author}
  {\bibfnamefont {M.}~\bibnamefont {Kieferov{\'{a}}}}, \bibinfo {author}
  {\bibfnamefont {I.~D.}\ \bibnamefont {Kivlichan}}, \bibinfo {author}
  {\bibfnamefont {T.}~\bibnamefont {Menke}}, \bibinfo {author} {\bibfnamefont
  {B.}~\bibnamefont {Peropadre}}, \bibinfo {author} {\bibfnamefont {N.~P.~D.}\
  \bibnamefont {Sawaya}}, \bibinfo {author} {\bibfnamefont {S.}~\bibnamefont
  {Sim}}, \bibinfo {author} {\bibfnamefont {L.}~\bibnamefont {Veis}},\ and\
  \bibinfo {author} {\bibfnamefont {A.}~\bibnamefont {Aspuru-Guzik}},\
  }\bibfield  {title} {\bibinfo {title} {Quantum chemistry in the age of
  quantum computing},\ }\href {https://doi.org/10.1021/acs.chemrev.8b00803}
  {\bibfield  {journal} {\bibinfo  {journal} {Chem. Rev.}\ }\textbf {\bibinfo
  {volume} {119}},\ \bibinfo {pages} {10856} (\bibinfo {year}
  {2019})}\BibitemShut {NoStop}%
\bibitem [{\citenamefont {McArdle}\ \emph {et~al.}(2020)\citenamefont
  {McArdle}, \citenamefont {Endo}, \citenamefont {Aspuru-Guzik}, \citenamefont
  {Benjamin},\ and\ \citenamefont {Yuan}}]{mcardle20_rev}%
  \BibitemOpen
  \bibfield  {author} {\bibinfo {author} {\bibfnamefont {S.}~\bibnamefont
  {McArdle}}, \bibinfo {author} {\bibfnamefont {S.}~\bibnamefont {Endo}},
  \bibinfo {author} {\bibfnamefont {A.}~\bibnamefont {Aspuru-Guzik}}, \bibinfo
  {author} {\bibfnamefont {S.~C.}\ \bibnamefont {Benjamin}},\ and\ \bibinfo
  {author} {\bibfnamefont {X.}~\bibnamefont {Yuan}},\ }\bibfield  {title}
  {\bibinfo {title} {Quantum computational chemistry},\ }\href
  {https://doi.org/10.1103/RevModPhys.92.015003} {\bibfield  {journal}
  {\bibinfo  {journal} {Rev. Mod. Phys.}\ }\textbf {\bibinfo {volume} {92}},\
  \bibinfo {pages} {015003} (\bibinfo {year} {2020})}\BibitemShut {NoStop}%
\bibitem [{\citenamefont {Huh}\ \emph {et~al.}(2015)\citenamefont {Huh},
  \citenamefont {Guerreschi}, \citenamefont {Peropadre}, \citenamefont
  {Mcclean},\ and\ \citenamefont {Aspuru-Guzik}}]{huh15}%
  \BibitemOpen
  \bibfield  {author} {\bibinfo {author} {\bibfnamefont {J.}~\bibnamefont
  {Huh}}, \bibinfo {author} {\bibfnamefont {G.~G.}\ \bibnamefont {Guerreschi}},
  \bibinfo {author} {\bibfnamefont {B.}~\bibnamefont {Peropadre}}, \bibinfo
  {author} {\bibfnamefont {J.~R.}\ \bibnamefont {Mcclean}},\ and\ \bibinfo
  {author} {\bibfnamefont {A.}~\bibnamefont {Aspuru-Guzik}},\ }\bibfield
  {title} {\bibinfo {title} {{Boson sampling for molecular vibronic spectra}},\
  }\href {https://doi.org/10.1038/nphoton.2015.153} {\bibfield  {journal}
  {\bibinfo  {journal} {Nat. Photonics}\ }\textbf {\bibinfo {volume} {9}},\
  \bibinfo {pages} {615} (\bibinfo {year} {2015})},\ \Eprint
  {https://arxiv.org/abs/1412.8427} {1412.8427} \BibitemShut {NoStop}%
\bibitem [{\citenamefont {Sawaya}\ and\ \citenamefont
  {Huh}(2019)}]{sawaya19_vibronic}%
  \BibitemOpen
  \bibfield  {author} {\bibinfo {author} {\bibfnamefont {N.~P.~D.}\
  \bibnamefont {Sawaya}}\ and\ \bibinfo {author} {\bibfnamefont
  {J.}~\bibnamefont {Huh}},\ }\bibfield  {title} {\bibinfo {title} {Quantum
  algorithm for calculating molecular vibronic spectra},\ }\href
  {https://doi.org/10.1021/acs.jpclett.9b01117} {\bibfield  {journal} {\bibinfo
   {journal} {J. Phys. Chem. Lett.}\ }\textbf {\bibinfo {volume} {10}},\
  \bibinfo {pages} {3586} (\bibinfo {year} {2019})}\BibitemShut {NoStop}%
\bibitem [{\citenamefont {Kosugi}\ and\ \citenamefont
  {Matsushita}(2020)}]{kosugi20_linresp_chargespin}%
  \BibitemOpen
  \bibfield  {author} {\bibinfo {author} {\bibfnamefont {T.}~\bibnamefont
  {Kosugi}}\ and\ \bibinfo {author} {\bibfnamefont {Y.-i.}\ \bibnamefont
  {Matsushita}},\ }\bibfield  {title} {\bibinfo {title} {Linear-response
  functions of molecules on a quantum computer: Charge and spin responses and
  optical absorption},\ }\href
  {https://doi.org/10.1103/PhysRevResearch.2.033043} {\bibfield  {journal}
  {\bibinfo  {journal} {Phys. Rev. Research}\ }\textbf {\bibinfo {volume}
  {2}},\ \bibinfo {pages} {033043} (\bibinfo {year} {2020})}\BibitemShut
  {NoStop}%
\bibitem [{\citenamefont {Cai}\ \emph {et~al.}(2020)\citenamefont {Cai},
  \citenamefont {Fang}, \citenamefont {Fan},\ and\ \citenamefont
  {Li}}]{cai20_qmolecresp}%
  \BibitemOpen
  \bibfield  {author} {\bibinfo {author} {\bibfnamefont {X.}~\bibnamefont
  {Cai}}, \bibinfo {author} {\bibfnamefont {W.-H.}\ \bibnamefont {Fang}},
  \bibinfo {author} {\bibfnamefont {H.}~\bibnamefont {Fan}},\ and\ \bibinfo
  {author} {\bibfnamefont {Z.}~\bibnamefont {Li}},\ }\bibfield  {title}
  {\bibinfo {title} {Quantum computation of molecular response properties},\
  }\href@noop {} {\bibfield  {journal} {\bibinfo  {journal} {{a}rXiv}\ }
  (\bibinfo {year} {2020})},\ \Eprint {https://arxiv.org/abs/arXiv:2001.03406}
  {arXiv:2001.03406} \BibitemShut {NoStop}%
\bibitem [{\citenamefont {Ibe}\ \emph {et~al.}(2020)\citenamefont {Ibe},
  \citenamefont {Nakagawa}, \citenamefont {Yamamoto}, \citenamefont {Mitarai},
  \citenamefont {Gao},\ and\ \citenamefont {Kobayashi}}]{ibe20_transition}%
  \BibitemOpen
  \bibfield  {author} {\bibinfo {author} {\bibfnamefont {Y.}~\bibnamefont
  {Ibe}}, \bibinfo {author} {\bibfnamefont {Y.~O.}\ \bibnamefont {Nakagawa}},
  \bibinfo {author} {\bibfnamefont {T.}~\bibnamefont {Yamamoto}}, \bibinfo
  {author} {\bibfnamefont {K.}~\bibnamefont {Mitarai}}, \bibinfo {author}
  {\bibfnamefont {Q.}~\bibnamefont {Gao}},\ and\ \bibinfo {author}
  {\bibfnamefont {T.}~\bibnamefont {Kobayashi}},\ }\bibfield  {title} {\bibinfo
  {title} {Calculating transition amplitudes by variational quantum
  eigensolvers},\ }\href@noop {} {\bibfield  {journal} {\bibinfo  {journal}
  {{a}rXiv}\ } (\bibinfo {year} {2020})},\ \Eprint
  {https://arxiv.org/abs/arXiv:2002.11724} {arXiv:2002.11724} \BibitemShut
  {NoStop}%
\bibitem [{\citenamefont {Sawaya}\ \emph {et~al.}(2021)\citenamefont {Sawaya},
  \citenamefont {Paesani},\ and\ \citenamefont {Tabor}}]{sawaya2021ir}%
  \BibitemOpen
  \bibfield  {author} {\bibinfo {author} {\bibfnamefont {N.~P.}\ \bibnamefont
  {Sawaya}}, \bibinfo {author} {\bibfnamefont {F.}~\bibnamefont {Paesani}},\
  and\ \bibinfo {author} {\bibfnamefont {D.~P.}\ \bibnamefont {Tabor}},\
  }\bibfield  {title} {\bibinfo {title} {Near-and long-term quantum algorithmic
  approaches for vibrational spectroscopy},\ }\href@noop {} {\bibfield
  {journal} {\bibinfo  {journal} {Physical Review A}\ }\textbf {\bibinfo
  {volume} {104}},\ \bibinfo {pages} {062419} (\bibinfo {year}
  {2021})}\BibitemShut {NoStop}%
\bibitem [{\citenamefont {Berne}\ and\ \citenamefont
  {Harp}(1970)}]{berne70_corrbook}%
  \BibitemOpen
  \bibfield  {author} {\bibinfo {author} {\bibfnamefont {B.~J.}\ \bibnamefont
  {Berne}}\ and\ \bibinfo {author} {\bibfnamefont {G.~D.}\ \bibnamefont
  {Harp}},\ }\bibinfo {title} {On the calculation of time correlation
  functions},\ in\ \href@noop {} {\emph {\bibinfo {booktitle} {Advances in
  Chemical Physics}}}\ (\bibinfo  {publisher} {John Wiley and Sons, Ltd},\
  \bibinfo {year} {1970})\ pp.\ \bibinfo {pages} {63--227}\BibitemShut
  {NoStop}%
\bibitem [{\citenamefont {Florencio}\ and\ \citenamefont
  {de~Alcantara~Bonfim}(2020)}]{florencio2020corrrev}%
  \BibitemOpen
  \bibfield  {author} {\bibinfo {author} {\bibfnamefont {J.}~\bibnamefont
  {Florencio}}\ and\ \bibinfo {author} {\bibfnamefont {O.}~\bibnamefont
  {de~Alcantara~Bonfim}},\ }\bibfield  {title} {\bibinfo {title} {Recent
  advances in the calculation of dynamical correlation functions},\ }\href@noop
  {} {\bibfield  {journal} {\bibinfo  {journal} {Frontiers in Physics}\
  }\textbf {\bibinfo {volume} {8}},\ \bibinfo {pages} {557277} (\bibinfo {year}
  {2020})}\BibitemShut {NoStop}%
\bibitem [{\citenamefont {Roggero}\ and\ \citenamefont
  {Carlson}(2019)}]{Roggero19_linresp}%
  \BibitemOpen
  \bibfield  {author} {\bibinfo {author} {\bibfnamefont {A.}~\bibnamefont
  {Roggero}}\ and\ \bibinfo {author} {\bibfnamefont {J.}~\bibnamefont
  {Carlson}},\ }\bibfield  {title} {\bibinfo {title} {Dynamic linear response
  quantum algorithm},\ }\href {https://doi.org/10.1103/PhysRevC.100.034610}
  {\bibfield  {journal} {\bibinfo  {journal} {Phys. Rev. C}\ }\textbf {\bibinfo
  {volume} {100}},\ \bibinfo {pages} {034610} (\bibinfo {year}
  {2019})}\BibitemShut {NoStop}%
\bibitem [{\citenamefont {Chen}\ \emph {et~al.}(2021)\citenamefont {Chen},
  \citenamefont {Nusspickel}, \citenamefont {Tilly}, \citenamefont {Booth}
  \emph {et~al.}}]{chen2021vardyncorrfuncs}%
  \BibitemOpen
  \bibfield  {author} {\bibinfo {author} {\bibfnamefont {H.}~\bibnamefont
  {Chen}}, \bibinfo {author} {\bibfnamefont {M.}~\bibnamefont {Nusspickel}},
  \bibinfo {author} {\bibfnamefont {J.}~\bibnamefont {Tilly}}, \bibinfo
  {author} {\bibfnamefont {G.~H.}\ \bibnamefont {Booth}}, \emph {et~al.},\
  }\bibfield  {title} {\bibinfo {title} {Variational quantum eigensolver for
  dynamic correlation functions},\ }\href@noop {} {\bibfield  {journal}
  {\bibinfo  {journal} {Physical Review A}\ }\textbf {\bibinfo {volume}
  {104}},\ \bibinfo {pages} {032405} (\bibinfo {year} {2021})}\BibitemShut
  {NoStop}%
\bibitem [{\citenamefont {Bravo-Prieto}\ \emph {et~al.}(2019)\citenamefont
  {Bravo-Prieto}, \citenamefont {LaRose}, \citenamefont {Cerezo}, \citenamefont
  {Subasi}, \citenamefont {Cincio},\ and\ \citenamefont
  {Coles}}]{bravo19_vqls}%
  \BibitemOpen
  \bibfield  {author} {\bibinfo {author} {\bibfnamefont {C.}~\bibnamefont
  {Bravo-Prieto}}, \bibinfo {author} {\bibfnamefont {R.}~\bibnamefont
  {LaRose}}, \bibinfo {author} {\bibfnamefont {M.}~\bibnamefont {Cerezo}},
  \bibinfo {author} {\bibfnamefont {Y.}~\bibnamefont {Subasi}}, \bibinfo
  {author} {\bibfnamefont {L.}~\bibnamefont {Cincio}},\ and\ \bibinfo {author}
  {\bibfnamefont {P.~J.}\ \bibnamefont {Coles}},\ }\href@noop {} {\bibinfo
  {title} {Variational quantum linear solver}} (\bibinfo {year} {2019}),\
  \Eprint {https://arxiv.org/abs/arXiv:1909.05820} {arXiv:1909.05820}
  \BibitemShut {NoStop}%
\bibitem [{\citenamefont {Gily{\'e}n}\ \emph {et~al.}(2019)\citenamefont
  {Gily{\'e}n}, \citenamefont {Su}, \citenamefont {Low},\ and\ \citenamefont
  {Wiebe}}]{gilyen2019qsvd}%
  \BibitemOpen
  \bibfield  {author} {\bibinfo {author} {\bibfnamefont {A.}~\bibnamefont
  {Gily{\'e}n}}, \bibinfo {author} {\bibfnamefont {Y.}~\bibnamefont {Su}},
  \bibinfo {author} {\bibfnamefont {G.~H.}\ \bibnamefont {Low}},\ and\ \bibinfo
  {author} {\bibfnamefont {N.}~\bibnamefont {Wiebe}},\ }\bibfield  {title}
  {\bibinfo {title} {Quantum singular value transformation and beyond:
  exponential improvements for quantum matrix arithmetics},\ }in\ \href@noop {}
  {\emph {\bibinfo {booktitle} {Proceedings of the 51st Annual ACM SIGACT
  Symposium on Theory of Computing}}}\ (\bibinfo {year} {2019})\ pp.\ \bibinfo
  {pages} {193--204}\BibitemShut {NoStop}%
\bibitem [{\citenamefont {Xu}\ \emph {et~al.}(2021)\citenamefont {Xu},
  \citenamefont {Sun}, \citenamefont {Endo}, \citenamefont {Li}, \citenamefont
  {Benjamin},\ and\ \citenamefont {Yuan}}]{xu19_vari_linalg}%
  \BibitemOpen
  \bibfield  {author} {\bibinfo {author} {\bibfnamefont {X.}~\bibnamefont
  {Xu}}, \bibinfo {author} {\bibfnamefont {J.}~\bibnamefont {Sun}}, \bibinfo
  {author} {\bibfnamefont {S.}~\bibnamefont {Endo}}, \bibinfo {author}
  {\bibfnamefont {Y.}~\bibnamefont {Li}}, \bibinfo {author} {\bibfnamefont
  {S.~C.}\ \bibnamefont {Benjamin}},\ and\ \bibinfo {author} {\bibfnamefont
  {X.}~\bibnamefont {Yuan}},\ }\bibfield  {title} {\bibinfo {title}
  {Variational algorithms for linear algebra},\ }\href@noop {} {\bibfield
  {journal} {\bibinfo  {journal} {Science Bulletin}\ }\textbf {\bibinfo
  {volume} {66}},\ \bibinfo {pages} {2181} (\bibinfo {year}
  {2021})}\BibitemShut {NoStop}%
\bibitem [{\citenamefont {Endo}\ \emph {et~al.}(2020)\citenamefont {Endo},
  \citenamefont {Sun}, \citenamefont {Li}, \citenamefont {Benjamin},\ and\
  \citenamefont {Yuan}}]{endo20_vari_genproc}%
  \BibitemOpen
  \bibfield  {author} {\bibinfo {author} {\bibfnamefont {S.}~\bibnamefont
  {Endo}}, \bibinfo {author} {\bibfnamefont {J.}~\bibnamefont {Sun}}, \bibinfo
  {author} {\bibfnamefont {Y.}~\bibnamefont {Li}}, \bibinfo {author}
  {\bibfnamefont {S.~C.}\ \bibnamefont {Benjamin}},\ and\ \bibinfo {author}
  {\bibfnamefont {X.}~\bibnamefont {Yuan}},\ }\bibfield  {title} {\bibinfo
  {title} {Variational quantum simulation of general processes},\ }\href
  {https://doi.org/10.1103/PhysRevLett.125.010501} {\bibfield  {journal}
  {\bibinfo  {journal} {Phys. Rev. Lett.}\ }\textbf {\bibinfo {volume} {125}},\
  \bibinfo {pages} {010501} (\bibinfo {year} {2020})}\BibitemShut {NoStop}%
\bibitem [{\citenamefont {Wang}\ \emph {et~al.}(2021)\citenamefont {Wang},
  \citenamefont {Song},\ and\ \citenamefont {Wang}}]{wang2021vqsvd}%
  \BibitemOpen
  \bibfield  {author} {\bibinfo {author} {\bibfnamefont {X.}~\bibnamefont
  {Wang}}, \bibinfo {author} {\bibfnamefont {Z.}~\bibnamefont {Song}},\ and\
  \bibinfo {author} {\bibfnamefont {Y.}~\bibnamefont {Wang}},\ }\bibfield
  {title} {\bibinfo {title} {Variational quantum singular value
  decomposition},\ }\href@noop {} {\bibfield  {journal} {\bibinfo  {journal}
  {Quantum}\ }\textbf {\bibinfo {volume} {5}},\ \bibinfo {pages} {483}
  (\bibinfo {year} {2021})}\BibitemShut {NoStop}%
\bibitem [{\citenamefont {Huang}\ \emph {et~al.}(2021)\citenamefont {Huang},
  \citenamefont {Bharti},\ and\ \citenamefont {Rebentrost}}]{huang2021near}%
  \BibitemOpen
  \bibfield  {author} {\bibinfo {author} {\bibfnamefont {H.-Y.}\ \bibnamefont
  {Huang}}, \bibinfo {author} {\bibfnamefont {K.}~\bibnamefont {Bharti}},\ and\
  \bibinfo {author} {\bibfnamefont {P.}~\bibnamefont {Rebentrost}},\ }\bibfield
   {title} {\bibinfo {title} {Near-term quantum algorithms for linear systems
  of equations with regression loss functions},\ }\href@noop {} {\bibfield
  {journal} {\bibinfo  {journal} {New Journal of Physics}\ }\textbf {\bibinfo
  {volume} {23}},\ \bibinfo {pages} {113021} (\bibinfo {year}
  {2021})}\BibitemShut {NoStop}%
\bibitem [{\citenamefont {Berry}(2014)}]{berry2014pde}%
  \BibitemOpen
  \bibfield  {author} {\bibinfo {author} {\bibfnamefont {D.~W.}\ \bibnamefont
  {Berry}},\ }\bibfield  {title} {\bibinfo {title} {High-order quantum
  algorithm for solving linear differential equations},\ }\href@noop {}
  {\bibfield  {journal} {\bibinfo  {journal} {Journal of Physics A:
  Mathematical and Theoretical}\ }\textbf {\bibinfo {volume} {47}},\ \bibinfo
  {pages} {105301} (\bibinfo {year} {2014})}\BibitemShut {NoStop}%
\bibitem [{\citenamefont {Arrazola}\ \emph {et~al.}(2019)\citenamefont
  {Arrazola}, \citenamefont {Kalajdzievski}, \citenamefont {Weedbrook},\ and\
  \citenamefont {Lloyd}}]{arrazola2019pde}%
  \BibitemOpen
  \bibfield  {author} {\bibinfo {author} {\bibfnamefont {J.~M.}\ \bibnamefont
  {Arrazola}}, \bibinfo {author} {\bibfnamefont {T.}~\bibnamefont
  {Kalajdzievski}}, \bibinfo {author} {\bibfnamefont {C.}~\bibnamefont
  {Weedbrook}},\ and\ \bibinfo {author} {\bibfnamefont {S.}~\bibnamefont
  {Lloyd}},\ }\bibfield  {title} {\bibinfo {title} {Quantum algorithm for
  nonhomogeneous linear partial differential equations},\ }\href@noop {}
  {\bibfield  {journal} {\bibinfo  {journal} {Physical Review A}\ }\textbf
  {\bibinfo {volume} {100}},\ \bibinfo {pages} {032306} (\bibinfo {year}
  {2019})}\BibitemShut {NoStop}%
\bibitem [{\citenamefont {Liu}\ \emph {et~al.}(2021{\natexlab{a}})\citenamefont
  {Liu}, \citenamefont {Kolden}, \citenamefont {Krovi}, \citenamefont
  {Loureiro}, \citenamefont {Trivisa},\ and\ \citenamefont
  {Childs}}]{liu2021pde}%
  \BibitemOpen
  \bibfield  {author} {\bibinfo {author} {\bibfnamefont {J.-P.}\ \bibnamefont
  {Liu}}, \bibinfo {author} {\bibfnamefont {H.~{\O}.}\ \bibnamefont {Kolden}},
  \bibinfo {author} {\bibfnamefont {H.~K.}\ \bibnamefont {Krovi}}, \bibinfo
  {author} {\bibfnamefont {N.~F.}\ \bibnamefont {Loureiro}}, \bibinfo {author}
  {\bibfnamefont {K.}~\bibnamefont {Trivisa}},\ and\ \bibinfo {author}
  {\bibfnamefont {A.~M.}\ \bibnamefont {Childs}},\ }\bibfield  {title}
  {\bibinfo {title} {Efficient quantum algorithm for dissipative nonlinear
  differential equations},\ }\href@noop {} {\bibfield  {journal} {\bibinfo
  {journal} {Proceedings of the National Academy of Sciences}\ }\textbf
  {\bibinfo {volume} {118}} (\bibinfo {year} {2021}{\natexlab{a}})}\BibitemShut
  {NoStop}%
\bibitem [{\citenamefont {Orus}\ \emph {et~al.}(2019)\citenamefont {Orus},
  \citenamefont {Mugel},\ and\ \citenamefont {Lizaso}}]{orus2019finance}%
  \BibitemOpen
  \bibfield  {author} {\bibinfo {author} {\bibfnamefont {R.}~\bibnamefont
  {Orus}}, \bibinfo {author} {\bibfnamefont {S.}~\bibnamefont {Mugel}},\ and\
  \bibinfo {author} {\bibfnamefont {E.}~\bibnamefont {Lizaso}},\ }\bibfield
  {title} {\bibinfo {title} {Quantum computing for finance: Overview and
  prospects},\ }\href@noop {} {\bibfield  {journal} {\bibinfo  {journal}
  {Reviews in Physics}\ }\textbf {\bibinfo {volume} {4}},\ \bibinfo {pages}
  {100028} (\bibinfo {year} {2019})}\BibitemShut {NoStop}%
\bibitem [{\citenamefont {Pistoia}\ \emph {et~al.}(2021)\citenamefont
  {Pistoia}, \citenamefont {Ahmad}, \citenamefont {Ajagekar}, \citenamefont
  {Buts}, \citenamefont {Chakrabarti}, \citenamefont {Herman}, \citenamefont
  {Hu}, \citenamefont {Jena}, \citenamefont {Minssen}, \citenamefont {Niroula}
  \emph {et~al.}}]{pistoia2021jpmc_finance}%
  \BibitemOpen
  \bibfield  {author} {\bibinfo {author} {\bibfnamefont {M.}~\bibnamefont
  {Pistoia}}, \bibinfo {author} {\bibfnamefont {S.~F.}\ \bibnamefont {Ahmad}},
  \bibinfo {author} {\bibfnamefont {A.}~\bibnamefont {Ajagekar}}, \bibinfo
  {author} {\bibfnamefont {A.}~\bibnamefont {Buts}}, \bibinfo {author}
  {\bibfnamefont {S.}~\bibnamefont {Chakrabarti}}, \bibinfo {author}
  {\bibfnamefont {D.}~\bibnamefont {Herman}}, \bibinfo {author} {\bibfnamefont
  {S.}~\bibnamefont {Hu}}, \bibinfo {author} {\bibfnamefont {A.}~\bibnamefont
  {Jena}}, \bibinfo {author} {\bibfnamefont {P.}~\bibnamefont {Minssen}},
  \bibinfo {author} {\bibfnamefont {P.}~\bibnamefont {Niroula}}, \emph
  {et~al.},\ }\bibfield  {title} {\bibinfo {title} {Quantum machine learning
  for finance {ICCAD} special session paper},\ }in\ \href@noop {} {\emph
  {\bibinfo {booktitle} {2021 IEEE/ACM International Conference On Computer
  Aided Design (ICCAD)}}}\ (\bibinfo {organization} {IEEE},\ \bibinfo {year}
  {2021})\ pp.\ \bibinfo {pages} {1--9}\BibitemShut {NoStop}%
\bibitem [{\citenamefont {Liu}\ \emph {et~al.}(2021{\natexlab{b}})\citenamefont
  {Liu}, \citenamefont {Arunachalam},\ and\ \citenamefont {Temme}}]{liu20_qml}%
  \BibitemOpen
  \bibfield  {author} {\bibinfo {author} {\bibfnamefont {Y.}~\bibnamefont
  {Liu}}, \bibinfo {author} {\bibfnamefont {S.}~\bibnamefont {Arunachalam}},\
  and\ \bibinfo {author} {\bibfnamefont {K.}~\bibnamefont {Temme}},\ }\bibfield
   {title} {\bibinfo {title} {A rigorous and robust quantum speed-up in
  supervised machine learning},\ }\href@noop {} {\bibfield  {journal} {\bibinfo
   {journal} {Nature Physics}\ }\textbf {\bibinfo {volume} {17}},\ \bibinfo
  {pages} {1013} (\bibinfo {year} {2021}{\natexlab{b}})}\BibitemShut {NoStop}%
\bibitem [{\citenamefont {Biamonte}\ \emph {et~al.}(2017)\citenamefont
  {Biamonte}, \citenamefont {Wittek}, \citenamefont {Pancotti}, \citenamefont
  {Rebentrost}, \citenamefont {Wiebe},\ and\ \citenamefont
  {Lloyd}}]{biamonte2017quantum}%
  \BibitemOpen
  \bibfield  {author} {\bibinfo {author} {\bibfnamefont {J.}~\bibnamefont
  {Biamonte}}, \bibinfo {author} {\bibfnamefont {P.}~\bibnamefont {Wittek}},
  \bibinfo {author} {\bibfnamefont {N.}~\bibnamefont {Pancotti}}, \bibinfo
  {author} {\bibfnamefont {P.}~\bibnamefont {Rebentrost}}, \bibinfo {author}
  {\bibfnamefont {N.}~\bibnamefont {Wiebe}},\ and\ \bibinfo {author}
  {\bibfnamefont {S.}~\bibnamefont {Lloyd}},\ }\bibfield  {title} {\bibinfo
  {title} {Quantum machine learning},\ }\href@noop {} {\bibfield  {journal}
  {\bibinfo  {journal} {Nature}\ }\textbf {\bibinfo {volume} {549}},\ \bibinfo
  {pages} {195} (\bibinfo {year} {2017})}\BibitemShut {NoStop}%
\bibitem [{\citenamefont {Garc{\'\i}a}\ \emph {et~al.}(2022)\citenamefont
  {Garc{\'\i}a}, \citenamefont {Cruz-Benito},\ and\ \citenamefont
  {Garc{\'\i}a-Pe{\~n}alvo}}]{garcia2022qmlrev}%
  \BibitemOpen
  \bibfield  {author} {\bibinfo {author} {\bibfnamefont {D.~P.}\ \bibnamefont
  {Garc{\'\i}a}}, \bibinfo {author} {\bibfnamefont {J.}~\bibnamefont
  {Cruz-Benito}},\ and\ \bibinfo {author} {\bibfnamefont {F.~J.}\ \bibnamefont
  {Garc{\'\i}a-Pe{\~n}alvo}},\ }\bibfield  {title} {\bibinfo {title}
  {Systematic literature review: Quantum machine learning and its
  applications},\ }\href@noop {} {\bibfield  {journal} {\bibinfo  {journal}
  {arXiv preprint arXiv:2201.04093}\ } (\bibinfo {year} {2022})}\BibitemShut
  {NoStop}%
\bibitem [{\citenamefont {Mitarai}\ and\ \citenamefont
  {Fujii}(2019)}]{mitarai19_newhadamtest}%
  \BibitemOpen
  \bibfield  {author} {\bibinfo {author} {\bibfnamefont {K.}~\bibnamefont
  {Mitarai}}\ and\ \bibinfo {author} {\bibfnamefont {K.}~\bibnamefont
  {Fujii}},\ }\bibfield  {title} {\bibinfo {title} {Methodology for replacing
  indirect measurements with direct measurements},\ }\href
  {https://doi.org/10.1103/PhysRevResearch.1.013006} {\bibfield  {journal}
  {\bibinfo  {journal} {Phys. Rev. Research}\ }\textbf {\bibinfo {volume}
  {1}},\ \bibinfo {pages} {013006} (\bibinfo {year} {2019})}\BibitemShut
  {NoStop}%
\bibitem [{\citenamefont {Huggins}\ \emph {et~al.}(2020)\citenamefont
  {Huggins}, \citenamefont {Lee}, \citenamefont {Baek}, \citenamefont
  {O'Gorman},\ and\ \citenamefont {Whaley}}]{huggins20_nonorthogonal}%
  \BibitemOpen
  \bibfield  {author} {\bibinfo {author} {\bibfnamefont {W.~J.}\ \bibnamefont
  {Huggins}}, \bibinfo {author} {\bibfnamefont {J.}~\bibnamefont {Lee}},
  \bibinfo {author} {\bibfnamefont {U.}~\bibnamefont {Baek}}, \bibinfo {author}
  {\bibfnamefont {B.}~\bibnamefont {O'Gorman}},\ and\ \bibinfo {author}
  {\bibfnamefont {K.~B.}\ \bibnamefont {Whaley}},\ }\bibfield  {title}
  {\bibinfo {title} {A non-orthogonal variational quantum eigensolver},\ }\href
  {https://doi.org/10.1088/1367-2630/ab867b} {\bibfield  {journal} {\bibinfo
  {journal} {New Journal of Physics}\ }\textbf {\bibinfo {volume} {22}},\
  \bibinfo {pages} {073009} (\bibinfo {year} {2020})}\BibitemShut {NoStop}%
\bibitem [{\citenamefont {Huang}\ \emph {et~al.}(2022)\citenamefont {Huang},
  \citenamefont {Cai}, \citenamefont {Li}, \citenamefont {Ge}, \citenamefont
  {Hou}, \citenamefont {Li}, \citenamefont {Liu}, \citenamefont {Shi},
  \citenamefont {Chen}, \citenamefont {Zheng} \emph
  {et~al.}}]{huang2022linearoptical}%
  \BibitemOpen
  \bibfield  {author} {\bibinfo {author} {\bibfnamefont {K.}~\bibnamefont
  {Huang}}, \bibinfo {author} {\bibfnamefont {X.}~\bibnamefont {Cai}}, \bibinfo
  {author} {\bibfnamefont {H.}~\bibnamefont {Li}}, \bibinfo {author}
  {\bibfnamefont {Z.-Y.}\ \bibnamefont {Ge}}, \bibinfo {author} {\bibfnamefont
  {R.}~\bibnamefont {Hou}}, \bibinfo {author} {\bibfnamefont {H.}~\bibnamefont
  {Li}}, \bibinfo {author} {\bibfnamefont {T.}~\bibnamefont {Liu}}, \bibinfo
  {author} {\bibfnamefont {Y.}~\bibnamefont {Shi}}, \bibinfo {author}
  {\bibfnamefont {C.}~\bibnamefont {Chen}}, \bibinfo {author} {\bibfnamefont
  {D.}~\bibnamefont {Zheng}}, \emph {et~al.},\ }\bibfield  {title} {\bibinfo
  {title} {Simulating linear optical properties of molecules on a
  superconducting quantum processor},\ }\href@noop {} {\bibfield  {journal}
  {\bibinfo  {journal} {arXiv preprint arXiv:2201.02426}\ } (\bibinfo {year}
  {2022})}\BibitemShut {NoStop}%
\bibitem [{\citenamefont {Garcia-Escartin}\ and\ \citenamefont
  {Chamorro-Posada}(2013)}]{garcia13_destrswap}%
  \BibitemOpen
  \bibfield  {author} {\bibinfo {author} {\bibfnamefont {J.~C.}\ \bibnamefont
  {Garcia-Escartin}}\ and\ \bibinfo {author} {\bibfnamefont {P.}~\bibnamefont
  {Chamorro-Posada}},\ }\bibfield  {title} {\bibinfo {title} {{SWAP} test and
  {Hong-Ou-Mandel} effect are equivalent},\ }\href
  {https://doi.org/10.1103/PhysRevA.87.052330} {\bibfield  {journal} {\bibinfo
  {journal} {Phys. Rev. A}\ }\textbf {\bibinfo {volume} {87}},\ \bibinfo
  {pages} {052330} (\bibinfo {year} {2013})}\BibitemShut {NoStop}%
\bibitem [{\citenamefont {Knill}\ \emph {et~al.}(2007)\citenamefont {Knill},
  \citenamefont {Ortiz},\ and\ \citenamefont {Somma}}]{knill07_hadtest}%
  \BibitemOpen
  \bibfield  {author} {\bibinfo {author} {\bibfnamefont {E.}~\bibnamefont
  {Knill}}, \bibinfo {author} {\bibfnamefont {G.}~\bibnamefont {Ortiz}},\ and\
  \bibinfo {author} {\bibfnamefont {R.~D.}\ \bibnamefont {Somma}},\ }\bibfield
  {title} {\bibinfo {title} {Optimal quantum measurements of expectation values
  of observables},\ }\href {https://doi.org/10.1103/PhysRevA.75.012328}
  {\bibfield  {journal} {\bibinfo  {journal} {Phys. Rev. A}\ }\textbf {\bibinfo
  {volume} {75}},\ \bibinfo {pages} {012328} (\bibinfo {year}
  {2007})}\BibitemShut {NoStop}%
\bibitem [{\citenamefont {Dob\ifmmode \check{s}\else
  \v{s}\fi{}\'{\i}\ifmmode~\check{c}\else \v{c}\fi{}ek}\ \emph
  {et~al.}(2007)\citenamefont {Dob\ifmmode \check{s}\else
  \v{s}\fi{}\'{\i}\ifmmode~\check{c}\else \v{c}\fi{}ek}, \citenamefont
  {Johansson}, \citenamefont {Shumeiko},\ and\ \citenamefont
  {Wendin}}]{dobsicek07_hadtest}%
  \BibitemOpen
  \bibfield  {author} {\bibinfo {author} {\bibfnamefont {M.}~\bibnamefont
  {Dob\ifmmode \check{s}\else \v{s}\fi{}\'{\i}\ifmmode~\check{c}\else
  \v{c}\fi{}ek}}, \bibinfo {author} {\bibfnamefont {G.}~\bibnamefont
  {Johansson}}, \bibinfo {author} {\bibfnamefont {V.}~\bibnamefont
  {Shumeiko}},\ and\ \bibinfo {author} {\bibfnamefont {G.}~\bibnamefont
  {Wendin}},\ }\bibfield  {title} {\bibinfo {title} {Arbitrary accuracy
  iterative quantum phase estimation algorithm using a single ancillary qubit:
  A two-qubit benchmark},\ }\href {https://doi.org/10.1103/PhysRevA.76.030306}
  {\bibfield  {journal} {\bibinfo  {journal} {Phys. Rev. A}\ }\textbf {\bibinfo
  {volume} {76}},\ \bibinfo {pages} {030306} (\bibinfo {year}
  {2007})}\BibitemShut {NoStop}%
\bibitem [{\citenamefont {McArdle}\ \emph
  {et~al.}(2019{\natexlab{a}})\citenamefont {McArdle}, \citenamefont {Jones},
  \citenamefont {Endo}, \citenamefont {Li}, \citenamefont {Benjamin},\ and\
  \citenamefont {Yuan}}]{mcardle19_ite}%
  \BibitemOpen
  \bibfield  {author} {\bibinfo {author} {\bibfnamefont {S.}~\bibnamefont
  {McArdle}}, \bibinfo {author} {\bibfnamefont {T.}~\bibnamefont {Jones}},
  \bibinfo {author} {\bibfnamefont {S.}~\bibnamefont {Endo}}, \bibinfo {author}
  {\bibfnamefont {Y.}~\bibnamefont {Li}}, \bibinfo {author} {\bibfnamefont
  {S.~C.}\ \bibnamefont {Benjamin}},\ and\ \bibinfo {author} {\bibfnamefont
  {X.}~\bibnamefont {Yuan}},\ }\bibfield  {title} {\bibinfo {title}
  {Variational ansatz-based quantum simulation of imaginary time evolution},\
  }\href {https://doi.org/10.1038/s41534-019-0187-2} {\bibfield  {journal}
  {\bibinfo  {journal} {npj Quantum Inf.}\ }\textbf {\bibinfo {volume} {5}},\
  \bibinfo {pages} {75} (\bibinfo {year} {2019}{\natexlab{a}})}\BibitemShut
  {NoStop}%
\bibitem [{\citenamefont {Ollitrault}\ \emph
  {et~al.}(2020{\natexlab{a}})\citenamefont {Ollitrault}, \citenamefont
  {Kandala}, \citenamefont {Chen}, \citenamefont {Barkoutsos}, \citenamefont
  {Mezzacapo}, \citenamefont {Pistoia}, \citenamefont {Sheldon}, \citenamefont
  {Woerner}, \citenamefont {Gambetta},\ and\ \citenamefont
  {Tavernelli}}]{ollitrault19_eom}%
  \BibitemOpen
  \bibfield  {author} {\bibinfo {author} {\bibfnamefont {P.~J.}\ \bibnamefont
  {Ollitrault}}, \bibinfo {author} {\bibfnamefont {A.}~\bibnamefont {Kandala}},
  \bibinfo {author} {\bibfnamefont {C.-F.}\ \bibnamefont {Chen}}, \bibinfo
  {author} {\bibfnamefont {P.~K.}\ \bibnamefont {Barkoutsos}}, \bibinfo
  {author} {\bibfnamefont {A.}~\bibnamefont {Mezzacapo}}, \bibinfo {author}
  {\bibfnamefont {M.}~\bibnamefont {Pistoia}}, \bibinfo {author} {\bibfnamefont
  {S.}~\bibnamefont {Sheldon}}, \bibinfo {author} {\bibfnamefont
  {S.}~\bibnamefont {Woerner}}, \bibinfo {author} {\bibfnamefont {J.~M.}\
  \bibnamefont {Gambetta}},\ and\ \bibinfo {author} {\bibfnamefont
  {I.}~\bibnamefont {Tavernelli}},\ }\bibfield  {title} {\bibinfo {title}
  {Quantum equation of motion for computing molecular excitation energies on a
  noisy quantum processor},\ }\href@noop {} {\bibfield  {journal} {\bibinfo
  {journal} {Physical Review Research}\ }\textbf {\bibinfo {volume} {2}},\
  \bibinfo {pages} {043140} (\bibinfo {year} {2020}{\natexlab{a}})}\BibitemShut
  {NoStop}%
\bibitem [{\citenamefont {McClean}\ \emph {et~al.}(2016)\citenamefont
  {McClean}, \citenamefont {Romero}, \citenamefont {Babbush},\ and\
  \citenamefont {Aspuru-Guzik}}]{mcclean16_njp}%
  \BibitemOpen
  \bibfield  {author} {\bibinfo {author} {\bibfnamefont {J.~R.}\ \bibnamefont
  {McClean}}, \bibinfo {author} {\bibfnamefont {J.}~\bibnamefont {Romero}},
  \bibinfo {author} {\bibfnamefont {R.}~\bibnamefont {Babbush}},\ and\ \bibinfo
  {author} {\bibfnamefont {A.}~\bibnamefont {Aspuru-Guzik}},\ }\bibfield
  {title} {\bibinfo {title} {The theory of variational hybrid quantum-classical
  algorithms},\ }\href {https://doi.org/10.1088/1367-2630/18/2/023023}
  {\bibfield  {journal} {\bibinfo  {journal} {New J. Phys.}\ }\textbf {\bibinfo
  {volume} {18}},\ \bibinfo {pages} {023023} (\bibinfo {year}
  {2016})}\BibitemShut {NoStop}%
\bibitem [{\citenamefont {Jones}\ \emph {et~al.}(2019)\citenamefont {Jones},
  \citenamefont {Endo}, \citenamefont {McArdle}, \citenamefont {Yuan},\ and\
  \citenamefont {Benjamin}}]{jones19_discovspectra}%
  \BibitemOpen
  \bibfield  {author} {\bibinfo {author} {\bibfnamefont {T.}~\bibnamefont
  {Jones}}, \bibinfo {author} {\bibfnamefont {S.}~\bibnamefont {Endo}},
  \bibinfo {author} {\bibfnamefont {S.}~\bibnamefont {McArdle}}, \bibinfo
  {author} {\bibfnamefont {X.}~\bibnamefont {Yuan}},\ and\ \bibinfo {author}
  {\bibfnamefont {S.~C.}\ \bibnamefont {Benjamin}},\ }\bibfield  {title}
  {\bibinfo {title} {Variational quantum algorithms for discovering
  {H}amiltonian spectra},\ }\href {https://doi.org/10.1103/PhysRevA.99.062304}
  {\bibfield  {journal} {\bibinfo  {journal} {Phys. Rev. A}\ }\textbf {\bibinfo
  {volume} {99}},\ \bibinfo {pages} {062304} (\bibinfo {year}
  {2019})}\BibitemShut {NoStop}%
\bibitem [{\citenamefont {Higgott}\ \emph {et~al.}(2019)\citenamefont
  {Higgott}, \citenamefont {Wang},\ and\ \citenamefont
  {Brierley}}]{higgott19_vqd}%
  \BibitemOpen
  \bibfield  {author} {\bibinfo {author} {\bibfnamefont {O.}~\bibnamefont
  {Higgott}}, \bibinfo {author} {\bibfnamefont {D.}~\bibnamefont {Wang}},\ and\
  \bibinfo {author} {\bibfnamefont {S.}~\bibnamefont {Brierley}},\ }\bibfield
  {title} {\bibinfo {title} {Variational quantum computation of excited
  states},\ }\href {https://doi.org/10.22331/q-2019-07-01-156} {\bibfield
  {journal} {\bibinfo  {journal} {Quantum}\ }\textbf {\bibinfo {volume} {3}},\
  \bibinfo {pages} {156} (\bibinfo {year} {2019})}\BibitemShut {NoStop}%
\bibitem [{\citenamefont {Motta}\ \emph {et~al.}(2019)\citenamefont {Motta},
  \citenamefont {Sun}, \citenamefont {Tan}, \citenamefont {O'Rourke},
  \citenamefont {Ye}, \citenamefont {Minnich}, \citenamefont {Brand{\~{a}}o},\
  and\ \citenamefont {Chan}}]{motta19_qite}%
  \BibitemOpen
  \bibfield  {author} {\bibinfo {author} {\bibfnamefont {M.}~\bibnamefont
  {Motta}}, \bibinfo {author} {\bibfnamefont {C.}~\bibnamefont {Sun}}, \bibinfo
  {author} {\bibfnamefont {A.~T.~K.}\ \bibnamefont {Tan}}, \bibinfo {author}
  {\bibfnamefont {M.~J.}\ \bibnamefont {O'Rourke}}, \bibinfo {author}
  {\bibfnamefont {E.}~\bibnamefont {Ye}}, \bibinfo {author} {\bibfnamefont
  {A.~J.}\ \bibnamefont {Minnich}}, \bibinfo {author} {\bibfnamefont {F.~G.
  S.~L.}\ \bibnamefont {Brand{\~{a}}o}},\ and\ \bibinfo {author} {\bibfnamefont
  {G.~K.-L.}\ \bibnamefont {Chan}},\ }\bibfield  {title} {\bibinfo {title}
  {Determining eigenstates and thermal states on a quantum computer using
  quantum imaginary time evolution},\ }\href
  {https://doi.org/10.1038/s41567-019-0704-4} {\bibfield  {journal} {\bibinfo
  {journal} {Nat. Phys.}\ }\textbf {\bibinfo {volume} {16}},\ \bibinfo {pages}
  {205} (\bibinfo {year} {2019})}\BibitemShut {NoStop}%
\bibitem [{\citenamefont {McClean}\ \emph {et~al.}(2017)\citenamefont
  {McClean}, \citenamefont {Kimchi-Schwartz}, \citenamefont {Carter},\ and\
  \citenamefont {de~Jong}}]{mclean17_qse}%
  \BibitemOpen
  \bibfield  {author} {\bibinfo {author} {\bibfnamefont {J.~R.}\ \bibnamefont
  {McClean}}, \bibinfo {author} {\bibfnamefont {M.~E.}\ \bibnamefont
  {Kimchi-Schwartz}}, \bibinfo {author} {\bibfnamefont {J.}~\bibnamefont
  {Carter}},\ and\ \bibinfo {author} {\bibfnamefont {W.~A.}\ \bibnamefont
  {de~Jong}},\ }\bibfield  {title} {\bibinfo {title} {Hybrid quantum-classical
  hierarchy for mitigation of decoherence and determination of excited
  states},\ }\href {https://doi.org/10.1103/PhysRevA.95.042308} {\bibfield
  {journal} {\bibinfo  {journal} {Phys. Rev. A}\ }\textbf {\bibinfo {volume}
  {95}},\ \bibinfo {pages} {042308} (\bibinfo {year} {2017})}\BibitemShut
  {NoStop}%
\bibitem [{\citenamefont {Santagati}\ \emph {et~al.}(2018)\citenamefont
  {Santagati}, \citenamefont {Wang}, \citenamefont {Gentile}, \citenamefont
  {Paesani}, \citenamefont {Wiebe}, \citenamefont {McClean}, \citenamefont
  {Morley-Short}, \citenamefont {Shadbolt}, \citenamefont {Bonneau},
  \citenamefont {Silverstone}, \citenamefont {Tew}, \citenamefont {Zhou},
  \citenamefont {O'Brien},\ and\ \citenamefont {Thompson}}]{santagati18_waves}%
  \BibitemOpen
  \bibfield  {author} {\bibinfo {author} {\bibfnamefont {R.}~\bibnamefont
  {Santagati}}, \bibinfo {author} {\bibfnamefont {J.}~\bibnamefont {Wang}},
  \bibinfo {author} {\bibfnamefont {A.~A.}\ \bibnamefont {Gentile}}, \bibinfo
  {author} {\bibfnamefont {S.}~\bibnamefont {Paesani}}, \bibinfo {author}
  {\bibfnamefont {N.}~\bibnamefont {Wiebe}}, \bibinfo {author} {\bibfnamefont
  {J.~R.}\ \bibnamefont {McClean}}, \bibinfo {author} {\bibfnamefont
  {S.}~\bibnamefont {Morley-Short}}, \bibinfo {author} {\bibfnamefont {P.~J.}\
  \bibnamefont {Shadbolt}}, \bibinfo {author} {\bibfnamefont {D.}~\bibnamefont
  {Bonneau}}, \bibinfo {author} {\bibfnamefont {J.~W.}\ \bibnamefont
  {Silverstone}}, \bibinfo {author} {\bibfnamefont {D.~P.}\ \bibnamefont
  {Tew}}, \bibinfo {author} {\bibfnamefont {X.}~\bibnamefont {Zhou}}, \bibinfo
  {author} {\bibfnamefont {J.~L.}\ \bibnamefont {O'Brien}},\ and\ \bibinfo
  {author} {\bibfnamefont {M.~G.}\ \bibnamefont {Thompson}},\ }\bibfield
  {title} {\bibinfo {title} {Witnessing eigenstates for quantum simulation of
  {Hamiltonian} spectra},\ }\href {https://doi.org/10.1126/sciadv.aap9646}
  {\bibfield  {journal} {\bibinfo  {journal} {Sci. Adv.}\ }\textbf {\bibinfo
  {volume} {4}},\ \bibinfo {pages} {eaap9646} (\bibinfo {year}
  {2018})}\BibitemShut {NoStop}%
\bibitem [{\citenamefont {Childs}\ and\ \citenamefont
  {Wiebe}(2012)}]{childs2012lcu}%
  \BibitemOpen
  \bibfield  {author} {\bibinfo {author} {\bibfnamefont {A.~M.}\ \bibnamefont
  {Childs}}\ and\ \bibinfo {author} {\bibfnamefont {N.}~\bibnamefont {Wiebe}},\
  }\bibfield  {title} {\bibinfo {title} {Hamiltonian simulation using linear
  combinations of unitary operations},\ }\href@noop {} {\bibfield  {journal}
  {\bibinfo  {journal} {arXiv preprint arXiv:1202.5822}\ } (\bibinfo {year}
  {2012})}\BibitemShut {NoStop}%
\bibitem [{\citenamefont {Nielsen}\ and\ \citenamefont
  {Chuang}(2011)}]{mikeike11}%
  \BibitemOpen
  \bibfield  {author} {\bibinfo {author} {\bibfnamefont {M.~A.}\ \bibnamefont
  {Nielsen}}\ and\ \bibinfo {author} {\bibfnamefont {I.~L.}\ \bibnamefont
  {Chuang}},\ }\href@noop {} {\emph {\bibinfo {title} {Quantum Computation and
  Quantum Information: 10th Anniversary Edition}}},\ \bibinfo {edition} {10th}\
  ed.\ (\bibinfo  {publisher} {Cambridge University Press},\ \bibinfo {address}
  {New York, NY, USA},\ \bibinfo {year} {2011})\BibitemShut {NoStop}%
\bibitem [{\citenamefont {Buhrman}\ \emph {et~al.}(2001)\citenamefont
  {Buhrman}, \citenamefont {Cleve}, \citenamefont {Watrous},\ and\
  \citenamefont {de~Wolf}}]{buhrman01_qfingerprint}%
  \BibitemOpen
  \bibfield  {author} {\bibinfo {author} {\bibfnamefont {H.}~\bibnamefont
  {Buhrman}}, \bibinfo {author} {\bibfnamefont {R.}~\bibnamefont {Cleve}},
  \bibinfo {author} {\bibfnamefont {J.}~\bibnamefont {Watrous}},\ and\ \bibinfo
  {author} {\bibfnamefont {R.}~\bibnamefont {de~Wolf}},\ }\bibfield  {title}
  {\bibinfo {title} {Quantum fingerprinting},\ }\href
  {https://doi.org/10.1103/PhysRevLett.87.167902} {\bibfield  {journal}
  {\bibinfo  {journal} {Phys. Rev. Lett.}\ }\textbf {\bibinfo {volume} {87}},\
  \bibinfo {pages} {167902} (\bibinfo {year} {2001})}\BibitemShut {NoStop}%
\bibitem [{\citenamefont {Havl{\'{\i}}{\v{c}}ek}\ \emph
  {et~al.}(2019)\citenamefont {Havl{\'{\i}}{\v{c}}ek}, \citenamefont
  {C{\'{o}}rcoles}, \citenamefont {Temme}, \citenamefont {Harrow},
  \citenamefont {Kandala}, \citenamefont {Chow},\ and\ \citenamefont
  {Gambetta}}]{havlek19_qml}%
  \BibitemOpen
  \bibfield  {author} {\bibinfo {author} {\bibfnamefont {V.}~\bibnamefont
  {Havl{\'{\i}}{\v{c}}ek}}, \bibinfo {author} {\bibfnamefont {A.~D.}\
  \bibnamefont {C{\'{o}}rcoles}}, \bibinfo {author} {\bibfnamefont
  {K.}~\bibnamefont {Temme}}, \bibinfo {author} {\bibfnamefont {A.~W.}\
  \bibnamefont {Harrow}}, \bibinfo {author} {\bibfnamefont {A.}~\bibnamefont
  {Kandala}}, \bibinfo {author} {\bibfnamefont {J.~M.}\ \bibnamefont {Chow}},\
  and\ \bibinfo {author} {\bibfnamefont {J.~M.}\ \bibnamefont {Gambetta}},\
  }\bibfield  {title} {\bibinfo {title} {Supervised learning with
  quantum-enhanced feature spaces},\ }\href
  {https://doi.org/10.1038/s41586-019-0980-2} {\bibfield  {journal} {\bibinfo
  {journal} {Nature}\ }\textbf {\bibinfo {volume} {567}},\ \bibinfo {pages}
  {209} (\bibinfo {year} {2019})}\BibitemShut {NoStop}%
\bibitem [{\citenamefont {Barenco}\ \emph {et~al.}(1995)\citenamefont
  {Barenco}, \citenamefont {Bennett}, \citenamefont {Cleve}, \citenamefont
  {DiVincenzo}, \citenamefont {Margolus}, \citenamefont {Shor}, \citenamefont
  {Sleator}, \citenamefont {Smolin},\ and\ \citenamefont
  {Weinfurter}}]{barenco95decomp}%
  \BibitemOpen
  \bibfield  {author} {\bibinfo {author} {\bibfnamefont {A.}~\bibnamefont
  {Barenco}}, \bibinfo {author} {\bibfnamefont {C.~H.}\ \bibnamefont
  {Bennett}}, \bibinfo {author} {\bibfnamefont {R.}~\bibnamefont {Cleve}},
  \bibinfo {author} {\bibfnamefont {D.~P.}\ \bibnamefont {DiVincenzo}},
  \bibinfo {author} {\bibfnamefont {N.}~\bibnamefont {Margolus}}, \bibinfo
  {author} {\bibfnamefont {P.}~\bibnamefont {Shor}}, \bibinfo {author}
  {\bibfnamefont {T.}~\bibnamefont {Sleator}}, \bibinfo {author} {\bibfnamefont
  {J.~A.}\ \bibnamefont {Smolin}},\ and\ \bibinfo {author} {\bibfnamefont
  {H.}~\bibnamefont {Weinfurter}},\ }\bibfield  {title} {\bibinfo {title}
  {Elementary gates for quantum computation},\ }\href@noop {} {\bibfield
  {journal} {\bibinfo  {journal} {Physical Review A}\ }\textbf {\bibinfo
  {volume} {52}},\ \bibinfo {pages} {3457} (\bibinfo {year}
  {1995})}\BibitemShut {NoStop}%
\bibitem [{\citenamefont {Crawford}\ \emph {et~al.}(2021)\citenamefont
  {Crawford}, \citenamefont {van Straaten}, \citenamefont {Wang}, \citenamefont
  {Parks}, \citenamefont {Campbell},\ and\ \citenamefont
  {Brierley}}]{crawford2021si}%
  \BibitemOpen
  \bibfield  {author} {\bibinfo {author} {\bibfnamefont {O.}~\bibnamefont
  {Crawford}}, \bibinfo {author} {\bibfnamefont {B.}~\bibnamefont {van
  Straaten}}, \bibinfo {author} {\bibfnamefont {D.}~\bibnamefont {Wang}},
  \bibinfo {author} {\bibfnamefont {T.}~\bibnamefont {Parks}}, \bibinfo
  {author} {\bibfnamefont {E.}~\bibnamefont {Campbell}},\ and\ \bibinfo
  {author} {\bibfnamefont {S.}~\bibnamefont {Brierley}},\ }\bibfield  {title}
  {\bibinfo {title} {Efficient quantum measurement of pauli operators in the
  presence of finite sampling error},\ }\href@noop {} {\bibfield  {journal}
  {\bibinfo  {journal} {Quantum}\ }\textbf {\bibinfo {volume} {5}},\ \bibinfo
  {pages} {385} (\bibinfo {year} {2021})}\BibitemShut {NoStop}%
\bibitem [{\citenamefont {Yen}\ \emph {et~al.}(2023)\citenamefont {Yen},
  \citenamefont {Ganeshram},\ and\ \citenamefont
  {Izmaylov}}]{yen2023deterministic}%
  \BibitemOpen
  \bibfield  {author} {\bibinfo {author} {\bibfnamefont {T.-C.}\ \bibnamefont
  {Yen}}, \bibinfo {author} {\bibfnamefont {A.}~\bibnamefont {Ganeshram}},\
  and\ \bibinfo {author} {\bibfnamefont {A.~F.}\ \bibnamefont {Izmaylov}},\
  }\bibfield  {title} {\bibinfo {title} {Deterministic improvements of quantum
  measurements with grouping of compatible operators, non-local
  transformations, and covariance estimates},\ }\href@noop {} {\bibfield
  {journal} {\bibinfo  {journal} {npj Quantum Information}\ }\textbf {\bibinfo
  {volume} {9}},\ \bibinfo {pages} {14} (\bibinfo {year} {2023})}\BibitemShut
  {NoStop}%
\bibitem [{\citenamefont {Vazquez}\ \emph {et~al.}(2022)\citenamefont
  {Vazquez}, \citenamefont {Hiptmair},\ and\ \citenamefont
  {Woerner}}]{vazquez2022qla_extrap}%
  \BibitemOpen
  \bibfield  {author} {\bibinfo {author} {\bibfnamefont {A.~C.}\ \bibnamefont
  {Vazquez}}, \bibinfo {author} {\bibfnamefont {R.}~\bibnamefont {Hiptmair}},\
  and\ \bibinfo {author} {\bibfnamefont {S.}~\bibnamefont {Woerner}},\
  }\bibfield  {title} {\bibinfo {title} {Enhancing the quantum linear systems
  algorithm using richardson extrapolation},\ }\href@noop {} {\bibfield
  {journal} {\bibinfo  {journal} {ACM Transactions on Quantum Computing}\
  }\textbf {\bibinfo {volume} {3}},\ \bibinfo {pages} {1} (\bibinfo {year}
  {2022})}\BibitemShut {NoStop}%
\bibitem [{\citenamefont {Richardson}\ and\ \citenamefont
  {Gaunt}(1927)}]{richardson1927deferred}%
  \BibitemOpen
  \bibfield  {author} {\bibinfo {author} {\bibfnamefont {L.~F.}\ \bibnamefont
  {Richardson}}\ and\ \bibinfo {author} {\bibfnamefont {J.~A.}\ \bibnamefont
  {Gaunt}},\ }\bibfield  {title} {\bibinfo {title} {Viii. the deferred approach
  to the limit},\ }\href@noop {} {\bibfield  {journal} {\bibinfo  {journal}
  {Philosophical Transactions of the Royal Society of London. Series A,
  containing papers of a mathematical or physical character}\ }\textbf
  {\bibinfo {volume} {226}},\ \bibinfo {pages} {299} (\bibinfo {year}
  {1927})}\BibitemShut {NoStop}%
\bibitem [{\citenamefont {Pozrikidis}(2008)}]{pozrikidis2008}%
  \BibitemOpen
  \bibfield  {author} {\bibinfo {author} {\bibfnamefont {C.}~\bibnamefont
  {Pozrikidis}},\ }\href@noop {} {\emph {\bibinfo {title} {Numerical
  computation in science and engineering}}}\ (\bibinfo  {publisher} {Oxford
  university press New York},\ \bibinfo {year} {2008})\BibitemShut {NoStop}%
\bibitem [{\citenamefont {Suzuki}(1985)}]{suzuki1985decomposition}%
  \BibitemOpen
  \bibfield  {author} {\bibinfo {author} {\bibfnamefont {M.}~\bibnamefont
  {Suzuki}},\ }\bibfield  {title} {\bibinfo {title} {Decomposition formulas of
  exponential operators and lie exponentials with some applications to quantum
  mechanics and statistical physics},\ }\href@noop {} {\bibfield  {journal}
  {\bibinfo  {journal} {Journal of mathematical physics}\ }\textbf {\bibinfo
  {volume} {26}},\ \bibinfo {pages} {601} (\bibinfo {year} {1985})}\BibitemShut
  {NoStop}%
\bibitem [{\citenamefont {Virtanen}\ \emph {et~al.}(2020)\citenamefont
  {Virtanen}, \citenamefont {Gommers}, \citenamefont {Oliphant}, \citenamefont
  {Haberland}, \citenamefont {Reddy}, \citenamefont {Cournapeau}, \citenamefont
  {Burovski}, \citenamefont {Peterson}, \citenamefont {Weckesser},
  \citenamefont {Bright} \emph {et~al.}}]{scipy}%
  \BibitemOpen
  \bibfield  {author} {\bibinfo {author} {\bibfnamefont {P.}~\bibnamefont
  {Virtanen}}, \bibinfo {author} {\bibfnamefont {R.}~\bibnamefont {Gommers}},
  \bibinfo {author} {\bibfnamefont {T.~E.}\ \bibnamefont {Oliphant}}, \bibinfo
  {author} {\bibfnamefont {M.}~\bibnamefont {Haberland}}, \bibinfo {author}
  {\bibfnamefont {T.}~\bibnamefont {Reddy}}, \bibinfo {author} {\bibfnamefont
  {D.}~\bibnamefont {Cournapeau}}, \bibinfo {author} {\bibfnamefont
  {E.}~\bibnamefont {Burovski}}, \bibinfo {author} {\bibfnamefont
  {P.}~\bibnamefont {Peterson}}, \bibinfo {author} {\bibfnamefont
  {W.}~\bibnamefont {Weckesser}}, \bibinfo {author} {\bibfnamefont
  {J.}~\bibnamefont {Bright}}, \emph {et~al.},\ }\bibfield  {title} {\bibinfo
  {title} {Scipy 1.0: fundamental algorithms for scientific computing in
  python},\ }\href@noop {} {\bibfield  {journal} {\bibinfo  {journal} {Nature
  methods}\ }\textbf {\bibinfo {volume} {17}},\ \bibinfo {pages} {261}
  (\bibinfo {year} {2020})}\BibitemShut {NoStop}%
\bibitem [{\citenamefont {Sawaya}(2022)}]{mat2qubit}%
  \BibitemOpen
  \bibfield  {author} {\bibinfo {author} {\bibfnamefont {N.~P.}\ \bibnamefont
  {Sawaya}},\ }\bibfield  {title} {\bibinfo {title} {mat2qubit: A lightweight
  pythonic package for qubit encodings of vibrational, bosonic, graph coloring,
  routing, scheduling, and general matrix problems},\ }\href@noop {} {\bibfield
   {journal} {\bibinfo  {journal} {arXiv preprint arXiv:2205.09776}\ }
  (\bibinfo {year} {2022})}\BibitemShut {NoStop}%
\bibitem [{\citenamefont {Sawaya}\ \emph {et~al.}(2022)\citenamefont {Sawaya},
  \citenamefont {Schmitz},\ and\ \citenamefont {Hadfield}}]{sawaya2022dqir}%
  \BibitemOpen
  \bibfield  {author} {\bibinfo {author} {\bibfnamefont {N.~P.}\ \bibnamefont
  {Sawaya}}, \bibinfo {author} {\bibfnamefont {A.~T.}\ \bibnamefont
  {Schmitz}},\ and\ \bibinfo {author} {\bibfnamefont {S.}~\bibnamefont
  {Hadfield}},\ }\bibfield  {title} {\bibinfo {title} {Encoding trade-offs and
  design toolkits in quantum algorithms for discrete optimization: coloring,
  routing, scheduling, and other problems},\ }\href@noop {} {\bibfield
  {journal} {\bibinfo  {journal} {arXiv preprint arXiv:2203.14432}\ } (\bibinfo
  {year} {2022})}\BibitemShut {NoStop}%
\bibitem [{\citenamefont {McClean}\ \emph {et~al.}(2020)\citenamefont
  {McClean}, \citenamefont {Rubin}, \citenamefont {Sung}, \citenamefont
  {Kivlichan}, \citenamefont {Bonet-Monroig}, \citenamefont {Cao},
  \citenamefont {Dai}, \citenamefont {Fried}, \citenamefont {Gidney},
  \citenamefont {Gimby} \emph {et~al.}}]{openfermion}%
  \BibitemOpen
  \bibfield  {author} {\bibinfo {author} {\bibfnamefont {J.~R.}\ \bibnamefont
  {McClean}}, \bibinfo {author} {\bibfnamefont {N.~C.}\ \bibnamefont {Rubin}},
  \bibinfo {author} {\bibfnamefont {K.~J.}\ \bibnamefont {Sung}}, \bibinfo
  {author} {\bibfnamefont {I.~D.}\ \bibnamefont {Kivlichan}}, \bibinfo {author}
  {\bibfnamefont {X.}~\bibnamefont {Bonet-Monroig}}, \bibinfo {author}
  {\bibfnamefont {Y.}~\bibnamefont {Cao}}, \bibinfo {author} {\bibfnamefont
  {C.}~\bibnamefont {Dai}}, \bibinfo {author} {\bibfnamefont {E.~S.}\
  \bibnamefont {Fried}}, \bibinfo {author} {\bibfnamefont {C.}~\bibnamefont
  {Gidney}}, \bibinfo {author} {\bibfnamefont {B.}~\bibnamefont {Gimby}}, \emph
  {et~al.},\ }\bibfield  {title} {\bibinfo {title} {Openfermion: the electronic
  structure package for quantum computers},\ }\href@noop {} {\bibfield
  {journal} {\bibinfo  {journal} {Quantum Science and Technology}\ }\textbf
  {\bibinfo {volume} {5}},\ \bibinfo {pages} {034014} (\bibinfo {year}
  {2020})}\BibitemShut {NoStop}%
\bibitem [{\citenamefont {McArdle}\ \emph
  {et~al.}(2019{\natexlab{b}})\citenamefont {McArdle}, \citenamefont {Mayorov},
  \citenamefont {Shan}, \citenamefont {Benjamin},\ and\ \citenamefont
  {Yuan}}]{mcardle19_qvibr}%
  \BibitemOpen
  \bibfield  {author} {\bibinfo {author} {\bibfnamefont {S.}~\bibnamefont
  {McArdle}}, \bibinfo {author} {\bibfnamefont {A.}~\bibnamefont {Mayorov}},
  \bibinfo {author} {\bibfnamefont {X.}~\bibnamefont {Shan}}, \bibinfo {author}
  {\bibfnamefont {S.}~\bibnamefont {Benjamin}},\ and\ \bibinfo {author}
  {\bibfnamefont {X.}~\bibnamefont {Yuan}},\ }\bibfield  {title} {\bibinfo
  {title} {Digital quantum simulation of molecular vibrations},\ }\href
  {https://doi.org/10.1039/c9sc01313j} {\bibfield  {journal} {\bibinfo
  {journal} {Chem. Sci.}\ }\textbf {\bibinfo {volume} {10}},\ \bibinfo {pages}
  {5725} (\bibinfo {year} {2019}{\natexlab{b}})}\BibitemShut {NoStop}%
\bibitem [{\citenamefont {Ollitrault}\ \emph
  {et~al.}(2020{\natexlab{b}})\citenamefont {Ollitrault}, \citenamefont
  {Baiardi}, \citenamefont {Reiher},\ and\ \citenamefont
  {Tavernelli}}]{ollitrault20_reiher_qvibr}%
  \BibitemOpen
  \bibfield  {author} {\bibinfo {author} {\bibfnamefont {P.~J.}\ \bibnamefont
  {Ollitrault}}, \bibinfo {author} {\bibfnamefont {A.}~\bibnamefont {Baiardi}},
  \bibinfo {author} {\bibfnamefont {M.}~\bibnamefont {Reiher}},\ and\ \bibinfo
  {author} {\bibfnamefont {I.}~\bibnamefont {Tavernelli}},\ }\bibfield  {title}
  {\bibinfo {title} {Hardware efficient quantum algorithms for vibrational
  structure calculations},\ }\href {https://doi.org/10.1039/d0sc01908a}
  {\bibfield  {journal} {\bibinfo  {journal} {Chem. Sci.}\ }\textbf {\bibinfo
  {volume} {11}},\ \bibinfo {pages} {6842} (\bibinfo {year}
  {2020}{\natexlab{b}})}\BibitemShut {NoStop}%
\bibitem [{\citenamefont {Huh}\ and\ \citenamefont
  {Yung}(2017)}]{huh17_manhong}%
  \BibitemOpen
  \bibfield  {author} {\bibinfo {author} {\bibfnamefont {J.}~\bibnamefont
  {Huh}}\ and\ \bibinfo {author} {\bibfnamefont {M.-H.}\ \bibnamefont {Yung}},\
  }\bibfield  {title} {\bibinfo {title} {Vibronic boson sampling: Generalized
  gaussian boson sampling for molecular vibronic spectra at finite
  temperature},\ }\href {https://doi.org/10.1038/s41598-017-07770-z} {\bibfield
   {journal} {\bibinfo  {journal} {Sci. Rep.}\ }\textbf {\bibinfo {volume}
  {7}},\ \bibinfo {pages} {7462} (\bibinfo {year} {2017})}\BibitemShut
  {NoStop}%
\bibitem [{\citenamefont {Jahangiri}\ \emph {et~al.}(2020)\citenamefont
  {Jahangiri}, \citenamefont {Arrazola}, \citenamefont {Quesada},\ and\
  \citenamefont {Delgado}}]{jahangiri2020xanaduvibronic}%
  \BibitemOpen
  \bibfield  {author} {\bibinfo {author} {\bibfnamefont {S.}~\bibnamefont
  {Jahangiri}}, \bibinfo {author} {\bibfnamefont {J.~M.}\ \bibnamefont
  {Arrazola}}, \bibinfo {author} {\bibfnamefont {N.}~\bibnamefont {Quesada}},\
  and\ \bibinfo {author} {\bibfnamefont {A.}~\bibnamefont {Delgado}},\
  }\bibfield  {title} {\bibinfo {title} {Quantum algorithm for simulating
  molecular vibrational excitations},\ }\href@noop {} {\bibfield  {journal}
  {\bibinfo  {journal} {Physical Chemistry Chemical Physics}\ }\textbf
  {\bibinfo {volume} {22}},\ \bibinfo {pages} {25528} (\bibinfo {year}
  {2020})}\BibitemShut {NoStop}%
\bibitem [{\citenamefont {Jnane}\ \emph {et~al.}(2021)\citenamefont {Jnane},
  \citenamefont {Sawaya}, \citenamefont {Peropadre}, \citenamefont
  {Aspuru-Guzik}, \citenamefont {Garcia-Patron},\ and\ \citenamefont
  {Huh}}]{jnane2021noncondon}%
  \BibitemOpen
  \bibfield  {author} {\bibinfo {author} {\bibfnamefont {H.}~\bibnamefont
  {Jnane}}, \bibinfo {author} {\bibfnamefont {N.~P.}\ \bibnamefont {Sawaya}},
  \bibinfo {author} {\bibfnamefont {B.}~\bibnamefont {Peropadre}}, \bibinfo
  {author} {\bibfnamefont {A.}~\bibnamefont {Aspuru-Guzik}}, \bibinfo {author}
  {\bibfnamefont {R.}~\bibnamefont {Garcia-Patron}},\ and\ \bibinfo {author}
  {\bibfnamefont {J.}~\bibnamefont {Huh}},\ }\bibfield  {title} {\bibinfo
  {title} {Analog quantum simulation of non-condon effects in molecular
  spectroscopy},\ }\href@noop {} {\bibfield  {journal} {\bibinfo  {journal}
  {ACS Photonics}\ }\textbf {\bibinfo {volume} {8}},\ \bibinfo {pages} {2007}
  (\bibinfo {year} {2021})}\BibitemShut {NoStop}%
\bibitem [{\citenamefont {Sawaya}\ \emph
  {et~al.}(2020{\natexlab{a}})\citenamefont {Sawaya}, \citenamefont {Menke},
  \citenamefont {Kyaw}, \citenamefont {Johri}, \citenamefont {Aspuru-Guzik},\
  and\ \citenamefont {Guerreschi}}]{sawaya20_dlev}%
  \BibitemOpen
  \bibfield  {author} {\bibinfo {author} {\bibfnamefont {N.~P.~D.}\
  \bibnamefont {Sawaya}}, \bibinfo {author} {\bibfnamefont {T.}~\bibnamefont
  {Menke}}, \bibinfo {author} {\bibfnamefont {T.~H.}\ \bibnamefont {Kyaw}},
  \bibinfo {author} {\bibfnamefont {S.}~\bibnamefont {Johri}}, \bibinfo
  {author} {\bibfnamefont {A.}~\bibnamefont {Aspuru-Guzik}},\ and\ \bibinfo
  {author} {\bibfnamefont {G.~G.}\ \bibnamefont {Guerreschi}},\ }\bibfield
  {title} {\bibinfo {title} {Resource-efficient digital quantum simulation of
  d-level systems for photonic, vibrational, and spin-s {Hamiltonian}},\
  }\href@noop {} {\bibfield  {journal} {\bibinfo  {journal} {npj Quantum Inf.}\
  }\textbf {\bibinfo {volume} {6}},\ \bibinfo {pages} {49} (\bibinfo {year}
  {2020}{\natexlab{a}})}\BibitemShut {NoStop}%
\bibitem [{\citenamefont {Sawaya}\ \emph
  {et~al.}(2020{\natexlab{b}})\citenamefont {Sawaya}, \citenamefont
  {Guerreschi},\ and\ \citenamefont {Holmes}}]{sawaya2020connectivity}%
  \BibitemOpen
  \bibfield  {author} {\bibinfo {author} {\bibfnamefont {N.~P.}\ \bibnamefont
  {Sawaya}}, \bibinfo {author} {\bibfnamefont {G.~G.}\ \bibnamefont
  {Guerreschi}},\ and\ \bibinfo {author} {\bibfnamefont {A.}~\bibnamefont
  {Holmes}},\ }\bibfield  {title} {\bibinfo {title} {On connectivity-dependent
  resource requirements for digital quantum simulation of d-level particles},\
  }in\ \href@noop {} {\emph {\bibinfo {booktitle} {2020 IEEE International
  Conference on Quantum Computing and Engineering (QCE)}}}\ (\bibinfo
  {organization} {IEEE},\ \bibinfo {year} {2020})\ pp.\ \bibinfo {pages}
  {180--190}\BibitemShut {NoStop}%
\bibitem [{\citenamefont {Harrow}\ \emph {et~al.}(2009)\citenamefont {Harrow},
  \citenamefont {Hassidim},\ and\ \citenamefont {Lloyd}}]{harrow2009_hhl}%
  \BibitemOpen
  \bibfield  {author} {\bibinfo {author} {\bibfnamefont {A.~W.}\ \bibnamefont
  {Harrow}}, \bibinfo {author} {\bibfnamefont {A.}~\bibnamefont {Hassidim}},\
  and\ \bibinfo {author} {\bibfnamefont {S.}~\bibnamefont {Lloyd}},\ }\bibfield
   {title} {\bibinfo {title} {Quantum algorithm for linear systems of
  equations},\ }\href {https://doi.org/10.1103/PhysRevLett.103.150502}
  {\bibfield  {journal} {\bibinfo  {journal} {Phys. Rev. Lett.}\ }\textbf
  {\bibinfo {volume} {103}},\ \bibinfo {pages} {150502} (\bibinfo {year}
  {2009})}\BibitemShut {NoStop}%
\bibitem [{\citenamefont {Bigoni}\ \emph {et~al.}(2016)\citenamefont {Bigoni},
  \citenamefont {Engsig-Karup},\ and\ \citenamefont
  {Marzouk}}]{bigoni2016tensortrain}%
  \BibitemOpen
  \bibfield  {author} {\bibinfo {author} {\bibfnamefont {D.}~\bibnamefont
  {Bigoni}}, \bibinfo {author} {\bibfnamefont {A.~P.}\ \bibnamefont
  {Engsig-Karup}},\ and\ \bibinfo {author} {\bibfnamefont {Y.~M.}\ \bibnamefont
  {Marzouk}},\ }\bibfield  {title} {\bibinfo {title} {Spectral tensor-train
  decomposition},\ }\href@noop {} {\bibfield  {journal} {\bibinfo  {journal}
  {SIAM Journal on Scientific Computing}\ }\textbf {\bibinfo {volume} {38}},\
  \bibinfo {pages} {A2405} (\bibinfo {year} {2016})}\BibitemShut {NoStop}%
\bibitem [{\citenamefont {Novikov}\ \emph {et~al.}(2020)\citenamefont
  {Novikov}, \citenamefont {Izmailov}, \citenamefont {Khrulkov}, \citenamefont
  {Figurnov},\ and\ \citenamefont {Oseledets}}]{novikov2020tensortrain}%
  \BibitemOpen
  \bibfield  {author} {\bibinfo {author} {\bibfnamefont {A.}~\bibnamefont
  {Novikov}}, \bibinfo {author} {\bibfnamefont {P.}~\bibnamefont {Izmailov}},
  \bibinfo {author} {\bibfnamefont {V.}~\bibnamefont {Khrulkov}}, \bibinfo
  {author} {\bibfnamefont {M.}~\bibnamefont {Figurnov}},\ and\ \bibinfo
  {author} {\bibfnamefont {I.~V.}\ \bibnamefont {Oseledets}},\ }\bibfield
  {title} {\bibinfo {title} {{Tensor Train Decomposition on TensorFlow
  (T3F)}},\ }\href@noop {} {\bibfield  {journal} {\bibinfo  {journal} {J. Mach.
  Learn. Res.}\ }\textbf {\bibinfo {volume} {21}},\ \bibinfo {pages} {1}
  (\bibinfo {year} {2020})}\BibitemShut {NoStop}%
\bibitem [{\citenamefont {Huang}\ \emph {et~al.}(2020)\citenamefont {Huang},
  \citenamefont {Kueng},\ and\ \citenamefont {Preskill}}]{huang2020shadow}%
  \BibitemOpen
  \bibfield  {author} {\bibinfo {author} {\bibfnamefont {H.-Y.}\ \bibnamefont
  {Huang}}, \bibinfo {author} {\bibfnamefont {R.}~\bibnamefont {Kueng}},\ and\
  \bibinfo {author} {\bibfnamefont {J.}~\bibnamefont {Preskill}},\ }\bibfield
  {title} {\bibinfo {title} {Predicting many properties of a quantum system
  from very few measurements},\ }\href@noop {} {\bibfield  {journal} {\bibinfo
  {journal} {Nature Physics}\ }\textbf {\bibinfo {volume} {16}},\ \bibinfo
  {pages} {1050} (\bibinfo {year} {2020})}\BibitemShut {NoStop}%
\bibitem [{\citenamefont {Seki}\ and\ \citenamefont
  {Yunoki}(2021)}]{seki21_power}%
  \BibitemOpen
  \bibfield  {author} {\bibinfo {author} {\bibfnamefont {K.}~\bibnamefont
  {Seki}}\ and\ \bibinfo {author} {\bibfnamefont {S.}~\bibnamefont {Yunoki}},\
  }\bibfield  {title} {\bibinfo {title} {Quantum power method by a
  superposition of time-evolved states},\ }\href
  {https://doi.org/10.1103/PRXQuantum.2.010333} {\bibfield  {journal} {\bibinfo
   {journal} {PRX Quantum}\ }\textbf {\bibinfo {volume} {2}},\ \bibinfo {pages}
  {010333} (\bibinfo {year} {2021})}\BibitemShut {NoStop}%
\bibitem [{\citenamefont {Verteletskyi}\ \emph {et~al.}(2020)\citenamefont
  {Verteletskyi}, \citenamefont {Yen},\ and\ \citenamefont
  {Izmaylov}}]{verteletskyi2020measurement}%
  \BibitemOpen
  \bibfield  {author} {\bibinfo {author} {\bibfnamefont {V.}~\bibnamefont
  {Verteletskyi}}, \bibinfo {author} {\bibfnamefont {T.-C.}\ \bibnamefont
  {Yen}},\ and\ \bibinfo {author} {\bibfnamefont {A.~F.}\ \bibnamefont
  {Izmaylov}},\ }\bibfield  {title} {\bibinfo {title} {Measurement optimization
  in the variational quantum eigensolver using a minimum clique cover},\
  }\href@noop {} {\bibfield  {journal} {\bibinfo  {journal} {The Journal of
  chemical physics}\ }\textbf {\bibinfo {volume} {152}},\ \bibinfo {pages}
  {124114} (\bibinfo {year} {2020})}\BibitemShut {NoStop}%
\bibitem [{\citenamefont {Cerezo}\ \emph {et~al.}(2021)\citenamefont {Cerezo},
  \citenamefont {Arrasmith}, \citenamefont {Babbush}, \citenamefont {Benjamin},
  \citenamefont {Endo}, \citenamefont {Fujii}, \citenamefont {McClean},
  \citenamefont {Mitarai}, \citenamefont {Yuan}, \citenamefont {Cincio} \emph
  {et~al.}}]{cerezo2021variational}%
  \BibitemOpen
  \bibfield  {author} {\bibinfo {author} {\bibfnamefont {M.}~\bibnamefont
  {Cerezo}}, \bibinfo {author} {\bibfnamefont {A.}~\bibnamefont {Arrasmith}},
  \bibinfo {author} {\bibfnamefont {R.}~\bibnamefont {Babbush}}, \bibinfo
  {author} {\bibfnamefont {S.~C.}\ \bibnamefont {Benjamin}}, \bibinfo {author}
  {\bibfnamefont {S.}~\bibnamefont {Endo}}, \bibinfo {author} {\bibfnamefont
  {K.}~\bibnamefont {Fujii}}, \bibinfo {author} {\bibfnamefont {J.~R.}\
  \bibnamefont {McClean}}, \bibinfo {author} {\bibfnamefont {K.}~\bibnamefont
  {Mitarai}}, \bibinfo {author} {\bibfnamefont {X.}~\bibnamefont {Yuan}},
  \bibinfo {author} {\bibfnamefont {L.}~\bibnamefont {Cincio}}, \emph
  {et~al.},\ }\bibfield  {title} {\bibinfo {title} {Variational quantum
  algorithms},\ }\href@noop {} {\bibfield  {journal} {\bibinfo  {journal}
  {Nature Reviews Physics}\ }\textbf {\bibinfo {volume} {3}},\ \bibinfo {pages}
  {625} (\bibinfo {year} {2021})}\BibitemShut {NoStop}%
\bibitem [{\citenamefont {Fedorov}\ \emph {et~al.}(2022)\citenamefont
  {Fedorov}, \citenamefont {Peng}, \citenamefont {Govind},\ and\ \citenamefont
  {Alexeev}}]{fedorov2022vqerev}%
  \BibitemOpen
  \bibfield  {author} {\bibinfo {author} {\bibfnamefont {D.~A.}\ \bibnamefont
  {Fedorov}}, \bibinfo {author} {\bibfnamefont {B.}~\bibnamefont {Peng}},
  \bibinfo {author} {\bibfnamefont {N.}~\bibnamefont {Govind}},\ and\ \bibinfo
  {author} {\bibfnamefont {Y.}~\bibnamefont {Alexeev}},\ }\bibfield  {title}
  {\bibinfo {title} {{VQE} method: A short survey and recent developments},\
  }\href@noop {} {\bibfield  {journal} {\bibinfo  {journal} {Materials Theory}\
  }\textbf {\bibinfo {volume} {6}},\ \bibinfo {pages} {1} (\bibinfo {year}
  {2022})}\BibitemShut {NoStop}%
\end{thebibliography}%

\end{document}